\documentclass{LMCS} 

\def\doi{9(1:03)2013}
\lmcsheading%
{\doi}
{1--34}
{}
{}
{Feb.~10, 2012}
{Feb.~15, 2013}
{}

\usepackage{enumerate}
\usepackage{hyperref}
\usepackage{yhmath}
\usepackage{subfig}

\usepackage{tikz}
\usetikzlibrary{arrows,positioning,chains,automata,scopes}
\usepackage{defs}

\usepackage{changebar}

\begin{document}

\title{Counting \texorpdfstring{\CTL}{CTL}}

\author[F.~Laroussinie]{Fran\c{c}ois Laroussinie\rsuper a}

\address{{\lsuper{a,c}}LIAFA, Universit\'e Paris Diderot -- Paris 7 \& CNRS UMR
  7089, France}
\email{\{Francois.Laroussinie, Eudes.Petonnet\}@liafa.jussieu.fr}

\author[A.~Meyer]{Antoine Meyer\rsuper b}

\address{{\lsuper b}LIGM, Universit\'e Paris Est -- Marne-la-Vall\'ee \& CNRS UMR
  8049, France}
\email{Antoine.Meyer@univ-mlv.fr}

\author[E.~Petonnet]{Eudes Petonnet\rsuper c}

\address{\vskip-6 pt}
%\email{Eudes.Petonnet@liafa.jussieu.fr}

\ACMCCS{[{\bf Software and its Engineering}]: Software creation and
  management---Software verification and validation---Formal software
  verification; [{\bf Theory of computation}]: Logic---Modal and
temporal logics \&  Verification by model checking}
\subjclass{F.4.1, D.2.4}

\keywords{branching time, counting, constraints, satisfiability,
  complexity}

\begin{abstract}
  This paper presents a range of quantitative extensions for the
  temporal logic \CTL. We enhance temporal modalities with the ability
  to constrain the number of states satisfying certain sub-formulas
  along paths. By selecting the combinations of Boolean and arithmetic
  operations allowed in constraints, one obtains several distinct
  logics generalizing \CTL. We provide a thorough analysis of their
  expressiveness and succinctness, and of the complexity of their
  model-checking and satisfiability problems (ranging from \P-complete
  to undecidable). Finally, we present two alternative logics with
  similar features and provide a comparative study of the properties
  of both variants.
\end{abstract}

\maketitle

%-----------------------------

\section{Introduction}

%\paragraph{Model checking  and temporal logic.}
Among the existing approaches to the formal verification of
automated systems, model checking \cite{clarke81b,queille82} aims at
automatically establishing the validity of a certain formal
specification (modeled as a formula in a suitable logic) over the
system under study (modeled for instance as a finite transition system).
This set of techniques is now well established and
successful, with several industrial applications.

To formalize the specification of temporal properties, for instance in
the case of reactive systems, temporal logics (TL) were proposed
thirty years ago~\cite{pnueli77} and widely studied since. They are
today used in many model-checking tools. There exists a wide variety
of temporal logics, differing for instance by the models over which
formulas are interpreted or by the kind of available temporal
modalities. Two well-known examples are \LTL in the linear-time
framework (where formulas are interpreted over infinite runs) and \CTL
for the branching-time case (where formulas are interpreted over
states of Kripke structures). See~\cite{emerson90} for a survey of
classical temporal logics for systems specification.

%\smallskip

%\paragraph{Extending temporal logics.} 
Temporal logics have been extended in various ways in order to
increase their expressive power.  For example, while \LTL and \CTL
only handle future operators, it is also possible to consider
past-time modalities to express properties of the past of a run. One
can also extend temporal logics with regular expressions (see for
instance ~\cite{wolper83b, emerson97b}). Other extensions were
proposed to handle \emph{quantitative} aspects of systems. For
example, some logics can contain timing constraints to specify that an
event, say $P_1$, has to occur less than $10$ time units before
another event $P_2$. This kind of temporal logics, such as
\TCTL~\cite{alur93, emerson92b}, have been particularly studied in the
framework of timed model checking.  Another quantitative extension
consists in \emph{probabilistic} logics where one can specify
probability bounds over the truth of some property (see for
instance~\cite{bianco95}).

%\smallskip

%\paragraph{Our  contribution.} 
We propose several extensions of \CTL with constraints over the number
of states satisfying certain sub-formulas along runs.  For example,
considering a model for an ATM, we can express the property ``whenever
the PIN is locked, at least three erroneous attempts have been made''
by: $\non \EFs{\nb \texttt{error} \leq 2} \texttt{lock}$ (one cannot
reach a state where the PIN is locked but less than two errors have
occurred).  Similarly, $\non \EFs{\nb \texttt{error} \geq 3}
\texttt{money}$ states that three mistakes forbid cash retrieval.  We
put a subscript on the temporal modality (as in \TCTL) to constrain
the runs over which the modality holds. Note that most properties of
this kind can also be expressed in \CTL\ by nesting $\Ex\_\U\_$
modalities, but the resulting formulas may be too large to be
conveniently handled by the user of a model checker. This is discussed
in more detail in Section \ref{sec:express}, where we study the
expressiveness of each of our fragments compared to \CTL. In some
cases, there exist natural translations into equivalent \CTL formulas,
implying that there is no strict gain in expressiveness. However,
these translations are often at best \emph{exponentially larger} than
the original formula. In other cases, we show that our extensions
strictly increases the expressive power of \CTL.

We consider the model checking problem for various sets $\calC$ of
constraints.
%, and denote by $\CCTL_\calC$ the corresponding extension
%of \CTL. 
We show that polynomial-time algorithms exist when considering Until
modalities with constraints of the form\footnote{Unless stated
  otherwise, complexity results always assume a binary encoding of
  constants.} $(\sum_i \nb \vfi_i) \sim c$ with $\sim\, \in
\{<,\leq,=,\geq,>\}$ and $c\in \Nat$. Additionally allowing Boolean
combinations of such constraints or integer coefficients in the sum
(or both) makes model checking $\DDP$-complete. We also consider the
case of ``diagonal'' constraints $(\nb\vfi-\nb\psi) \sim c$ and their
more general form $(\sum_i \pm \nb \vfi_i) \sim c$ with $c \in
\mathbb{Z}$ and show that model checking can still be done in
polynomial time. However, allowing Boolean combinations of such
constraints leads to undecidability. We also investigate the
complexity of the satisfiability problem, which is 2-\EXPTIME-complete
for all fragments without subtraction and undecidable
otherwise. Finally, in order to investigate alternative definitions of
counting logics generalizing \CTL, we define another semantics for our
logics (called cumulative semantics) and a logic with explicit
variables. In both cases, we show that it induces a complexity blow-up
for model checking, which becomes \PSPACE-complete without subtraction
and undecidable otherwise. The asymptotic complexity of satisfiability
remains however 2-\EXPTIME-complete in all decidable cases.

% \paragraph{Comparison with TL+past.}
% Our extension allows us to express properties that could be expressed
% with past-time operators. It is well known that such operators make
% specifications easier to write and more natural. For example, the
% property ``an alarm is always preceded by a problem'' can easily be
% stated with $\AG (\texttt{alarm} \impl \Fm \texttt{problem})$, and if
% we only have future-time modality we have to use a more complex
% formula. Adding past modalities to \CTL\ increases the complexity of
% model checking (for $\CTL+\Fm$ it becomes \PSPACE-complete~\cite{}).
% And the ability to nest past modalities allows us to specify
% properties over the ordering of events on the past.  On the contrary
% this cannot be done easily with our formalism that can only mention
% the number a formula has been true in the past\ldots But the use of
% integer gives succinctness: specifying that three mistakes are
% necessary to lock the PIN, requires in TL with past to nest three
% $\Fm$ modalities\ldots Moreover in our case we still have efficient
% model-checking algorithm.

%\smallskip

%\paragraph{Related work.}
Several existing works provide related results.
In~\cite{lar10}, we presented a preliminary version of the current
article. Proofs and constructions were since considerably refined, and
are provided here in greater detail. This paper also provides new
satisfiability results.
In~\cite{lmp10ltl}, we provided a similar study of counting extensions
of \LTL and \CTLs.
In~\cite{emerson97b}, an extension of \LTL with a kind of regular
expressions containing quantitative constraints over the number of
occurrences of sub-expressions is presented. This extension yields
algorithms whose time complexity is exponential in the size of
formulas and the \emph{value} of integer constants.
In~\cite{emerson99}, extensions of \CTL including parameters in
constraints are defined. One of these formalisms, namely {\sf GPCTL},
allows one to express properties with constraints defined as positive
Boolean combinations of sums of the form $\sum_i P_i \leq c$ where
every $P_i$ is an atomic proposition. Model-checking $\Ex\_\U\_$
formulas with such a constraint is shown to be \NP-complete and a
polynomial algorithm is given for a restricted logic (with
parameters).
Another interesting specification language is Sugar/PSL~\cite{psl03},
which defines many additional operators above \LTL and \CTLs. These
include in particular a kind of counting constraints used together
with regular expressions, but to our knowledge, there is no accurate
study of lower complexity bounds for these extensions~\cite{bfh05}.
In~\cite{yang97}, a branching-time temporal logic with general
counting constraints (using a variant of freeze variables) is defined
to specify event-driven real-time systems. To obtain decidability, the
authors restrict their analysis to systems verifying some bounded
progress condition.
%  an extension of \CTL adding diagonal counting
% constraints is defined, in a more general framework than ours. In
% particular, it uses registers, which are basically the same as freeze
% variables, in addition to discrete counting. To obtain decidability,
% they restrict the use of these diagonal constraints to events with
% bounded differences.
%
In~\cite{bouajjani95b,bouajjani95}, extensions of \LTL and \CTL with
Presburger constraints over the number of states satisfying a
formula are considered, for a class of infinite state processes.
The complexity of these problems is much higher than the cases we are
concerned with. Finally there also exist timed extensions of \CTL
interpreted over Kripke structures (see for instance~\cite{emerson92b}).

%\smallskip

%\paragraph{Plan of the paper.}
The paper is organized as follows.  In Section~\ref{sec:def}, we
introduce the definitions of the main formalisms we will use. In
Section~\ref{sec:express}, we show that several of our proposed extensions
are not more expressive than classical \CTL, yet exponentially
more succinct.  In Section~\ref{sec:mc}, we address the model-checking
problem and provide exact complexity results for almost all the logics we
introduce. In Section~\ref{sec:satisf} we study the complexity of the
satisfiability problem. Finally we present in Section~\ref{sec:exten} a different
logic with explicit counting variables, as well as an alternative semantics for
  our logics, together with the complexity of the related model-checking
problems.

\section{Definitions}
\label{sec:def}

\subsection{Models}

Let $\AP$ be a set of atomic propositions.  In branching-time temporal
logics, formulas are generally interpreted over states of Kripke structures.

\begin{defi}
  A \emph{Kripke structure} (or KS) $\calS$ is a tuple
  $\tuple{Q,R,\ell}$ where $Q$ is a finite set of states, $R\subseteq
  Q \times Q$ is a total\footnotemark{} transition relation and
  $\ell: Q \fleche 2^\AP$ is a labelling of states with atomic
  propositions.
\end{defi}

\footnotetext{By \emph{total} relation, we mean a relation $R
  \subseteq Q \times Q$ such that $\forall p \in Q, \exists q \in Q,
  (p,q) \in R$.}

A run $\rho$ of $\calS$ is an infinite sequence of states $q_0 q_1 q_2
\ldots$ such that\ $(q_i,q_{i+1})\in R$ for every $i$. We use
$\rho(i)$ to denote state $q_i$, $\rho_{|_i}$ to denote the prefix
$q_0 \cdots q_i$ of $\rho$, and $\epsilon$ to represent the empty
prefix. Notice that $\rho_{|_{-1}} = \epsilon$, but $\rho_{|_0} = q_0
\neq \epsilon$. $\Exec(q)$ denotes the set of runs starting from some
state $q \in Q$ and $\Exec(\calS)$ (resp.\ $\Execf(\calS)$) the set of
all runs (resp. finite prefixes of runs) in $\calS$. The length
$|\sigma|$ of a finite run prefix $\sigma$ is defined as usual
(i.e. $|\sigma| = 0$ if $\sigma = \epsilon$ and $|\sigma| = i+1$ if
$\sigma = q_0 \ldots q_i$). Note in particular that for any run
$\rho$, $|\rho_{|_i}| = i+1$.
We write $\sigma \leq \rho$ when $\sigma$ is a prefix of $\rho$.

We will also consider \emph{durational Kripke structures} (DKS), where
an integer duration is associated with every transition. A DKS $\calS
= \tuple{Q,R,\ell}$ is defined similarly to a KS, except that
$R\subseteq Q \times \mathbb{Z} \times Q$. The duration of a
transition is also called a \emph{weight} or a \emph{cost}, especially
when negative values are used to label a transition. We use $\DKSo$ to
denote the class of DKS in which every weight is $1$, $\DKSzo$ when
the weights belong to $\{0,1\}$, and $\DKSozo$ when they belong to
$\{-1,0,1\}$. The notion of weight is additively extended to finite
runs of DKS. The existence of a transition of weight $k$ between
states $p$ and $q$ is sometimes denoted as $p \lra_R^k q$, that of a
finite run of weight $k$ as $p \Lra_R^k q$. $R$ may be omitted when it
is clear from the context. The weight of a finite run $\rho$ is also
denoted as $\wght{\rho}$.

\subsection{Counting \texorpdfstring{\CTL}{CTL}}

We define several extensions of \CTL able to express constraints over the
number of times certain sub-formulas are satisfied along a run. 
% First % we introduce $\CCTL$:

\begin{defi}
  \label{def:cctl}
  Given a set of atomic propositions $\AP$, we define the logic
  $\CCTL$ as the set of formulas
  \[
  \vfi,\psi \grameg P \gramou \vfi \et \psi \gramou \non \vfi \gramou
  \Ex \vfi \Us{C} \psi \gramou \All \vfi \Us{C} \psi
  \]
  where $P\in \AP$ and $C$ is a constraint of the form
  \[
  C  \grameg  (\sum_{i=1}^{m}\alpha_i\cdot\nb\vfi_i) \sim k
  \]
  where $\vfi_i \in \CCTL$, $\alpha_i, k \in \Nat$ and $\sim\
  \in\{<,\leq,=,\geq,>\}$.
\end{defi}

We make use of the standard abbreviations $\ou, \impl, \equivaut,
\bot, \top$, as well as the additional modalities $\EFs{C} \vfi \egdef
\Ex \top \Us{C} \vfi$, $\AFs{C} \vfi \egdef \All \top \Us{C} \vfi$,
and their duals $\AGs{C} \vfi \egdef \non \EFs{C} \non \vfi$
and $\EGs{C} \vfi  \egdef \non \AFs{C} \non \vfi$.
Any formula occurring in a constraint $C$ associated with a modality
in $\Phi$ is considered a sub-formula of $\Phi$. The size $|\Phi|$
of $\Phi$ thus takes into account the size of these constraints and
their sub-formulas, assuming that integer constants are encoded in
\emph{binary} (unless explicitly stated otherwise). The DAG-size of
$\Phi$ is the number of distinct sub-formulas of $\Phi$.  As
model-checking algorithms may be implemented in such a way that the
truth value of each sub-formula is computed only once, for instance
using dynamic programming, this is generally more relevant to the
complexity of model-checking.

We also introduce several variants and extensions of \CCTL: 
\begin{iteMize}{$\bullet$}
\item $\CCTLa$ is the restriction of \CCTL where every coefficient
  $\alpha_i$ occurring in the constraints equals $1$.  Thus the constraints
  are of the form $(\sum_{i} \nb\vfi_i)\sim k$. For example, 
  $\EFs{\nb P + \nb P' = 10} P'' $ belongs to $\CCTL_1$.
\item $\CCTLpm$ is an extension of \CCTL with coefficients
  $\alpha_i$ in $\mathbb{Z}$. The formula $\EFs{\nb P - 3\cdot\nb
    P'=10} P''$ belongs to $\CCTLpm$.
\item $\CCTLb$ extends \CCTL by allowing Boolean combinations in the
  constraints. For example, $\EFs{\nb P < 4 \et \nb P' > 8}$ is in \CCTLb.\smallskip
\end{iteMize}

\noindent We can combine the previous variants and define the logics \CCTLpma,
\CCTLba, \CCTLbpm and \CCTLbpma. The semantics of our logics are
defined over Kripke structures as follows:

\begin{defi}
  \label{def:cctl-sem}
  The following clauses define the conditions for a state $q$ of some
  KS $\calS=\tuple{Q,R,\ell}$ to satisfy a formula $\vfi$ (written $q
  \sat_\calS \vfi$) by induction over the structure of $\vfi$ :
  \begin{alignat*}{2}
    & q \sat_\calS P & & \mbox{\quad iff \quad} P \in \ell(q) 
    \\
    & q \sat_\calS \lnot \vfi & & \mbox{\quad iff \quad} q \not\sat_\calS \vfi
    \\
    & q \sat_\calS \vfi \lor \psi & & \mbox{\quad iff \quad} q \sat_\calS \vfi
    \text{ or } q \sat_\calS \psi
    \\
    & q \sat_\calS \Ex \vfi \Us{C} \psi & & \mbox{\quad iff \quad} \exists \rho
    \in \Exec(q), \rho \sat_\calS \vfi \Us{C} \psi
    \\
    & q \sat_\calS \All \vfi \Us{C} \psi & & \mbox{\quad iff \quad} \forall \rho
    \in \Exec(q), \rho \sat_\calS \vfi \Us{C} \psi
  \end{alignat*}
  where $\rho \sat_\calS \vfi \Us{C} \psi \mbox{ iff } \exists i\geq 0, \
    \rho(i) \sat_\calS \psi, \; \rho_{|i-1} \sat_\calS C \mbox{ and }
    \forall 0 \leq j < i, \: \rho(j) \sat_\calS \vfi$.

    For every finite run prefix $\sigma = q_0 \ldots q_i$, the meaning
    of $\sigma \sat_\calS C$ is based on the interpretation of
    $\nb\vfi$ over $\sigma$, which is the number of states among $q_0,
    \ldots, q_i$ verifying $\vfi$, denoted by $ \csize{\sigma}{\vfi}$
    and defined as:
    $\csize{\sigma}{\vfi} \egdef |\{ j \mid 0 \leq j \leq i \et
    \sigma(j) \sat_\calS \vfi \}|$.
    Given these values, $C$ is evaluated as an ordinary equation or
    inequation over integer expressions.
\end{defi}

In the following we omit the subscript $\calS$ for $\sat$ when no
confusion occurs.  We use $\equiv$ to denote the standard
equivalence between formulas.

\begin{rem}
  \label{rem:g-semantics}
  It can be derived from the above definitions that formula $\EFs{C}
  \vfi$ holds from $q$ if and only if there is a run $\rho$ from $q$
  and an index $i$ such that $\rho(i) \sat \vfi$ and $\rho_{|_i-1|}
  \sat C$. Similarly, $\EGs{C} \vfi$ holds if and only if there exists
  a run $\rho$ such that, whenever a finite prefix of $\rho$ satisfies
  $C$, the next state must satisfy $\vfi$ (in other words, for all $i
  \geq 0$, $\rho_{|_{i-1}} \sat C \implies \rho(i) \sat \varphi$).
\end{rem}

\begin{rem}
  \label{rem:strict-prefix}
  The above semantics imply that the truth value of a constraint only
  depends on the \emph{strict} prefix of the run leading to (but not
  including) the current state. This is not an essential feature, and
  another definition would also be valid.
  However, this choice is consistent with the semantics of existing
  logics (in particular \TCTL \cite{alur93}). It also allows us to
  express the classical \X\ (or \emph{next}) operator as $\Ex \X \vfi =
  \Ex \Fs{\nb\top=1} \vfi$. Moreover, under this semantics the
  formulas $\Ex \vfi \U \psi$, $\Ex \top \Us{\nb \lnot \vfi = 0} \psi$
  and $\Ex \Fs{\nb \lnot \vfi = 0} \psi$ are all equivalent.
\end{rem}

\begin{rem}
  \label{rem:until}
  In all logics allowing Boolean connectives inside constraints, the
  modality $\F$ is sufficient to define $\U$. Indeed, $\Ex \phi \Us{C}
  \psi \equiv \Ex \Fs{C \et \nb(\non\phi)=0} \psi$ (and similarly for
  $\All$-quantified formulas). Thus every such logic can also be built
  from atomic propositions using Boolean operators and modalities
  $\EFs{C}\vfi$ and $\AFs{C}\vfi$ (or $\EGs{C}\vfi$).  Note that all
  these translations are succinct (linear in the size of formulas) and
  thus do not have any impact on complexity results.
\end{rem}

\begin{rem}
  The related temporal logic \TCTL, whose semantics is defined over
  \emph{timed} models (in particular durational Kripke structures),
  allows one to label temporal modalities with duration
  constraints. For instance, one may write $\All \varphi \U_{< k}
  \psi$ to express the fact that $\varphi$ is consistently true until,
  before $k$ time units have elapsed, $\psi$ eventually holds.

  When all transitions in a DKS have duration 1 (i.e. the duration of
  any run is equal to its length), \TCTL (or \RTCTL\
  in~\cite{emerson92b}) formulas can be directly expressed in any
  variant of \CCTL using only the sub-formula $\top$ inside
  constraints. A similar coding is also possible when one uses a
  proposition \emph{tick} to mark the elapse of time as
  in~\cite{lst-TCS2001}.
\end{rem}

%\subsection{\texorpdfstring{\TCTL}{TCTL}}

\subsection{Examples of \texorpdfstring{\CCTL}{CCTL} formulas.} 
\label{sec:examples-cctl}

We now give several examples of natural quantitative properties that
can be easily expressed with \CCTL-like logics.

\begin{enumerate}[(1)]
\item First consider an engine or plant that has to be controlled
  every 10000 cycles. Suppose a warning is activated whenever the
  number of elapsed cycles since the last control belongs to the
  interval $[9900;9950]$, and is maintained until the next control is
  done. Moreover, an alarm is raised when the number of cycles is
  above 10100 (unless a control was performed in-between) and is
  maintained until the next control. Such a specification could be
  expressed in \CCTL as follows:
  \begin{enumerate}[(a)]
  \item Either a control or a warning must occur in every period of
    9950 cycles:
    \[
    \AG \big( \AFs{\nb \texttt{cycle} \leq 9950} (\texttt{control} \:
    \ou \: \texttt{warning}) \big)
    \]
    where \texttt{cycle} (resp.\ \texttt{warning}, \texttt{control})
    labels states corresponding to the end of a cycle (resp.\ a
    warning, a control action).

  \item A warning cannot occur before 9900 cycles after a control: 
    \[
    \AG \big( \texttt{control} \;\impl\; \non \EFs{\nb \texttt{cycle}
      < 9900} \texttt{warning} \big).
    \]

  \item A control or an alarm occurs in every period of 10100 cycles:
    \[
    \AG \big( \AFs{\nb \texttt{cycle} \leq 10100} (\texttt{control} \:
    \ou \: \texttt{warning}) \big).
    \]

  \item An alarm cannot occur strictly before 10100 cycles after a control: 
    \[
    \AG \big( \texttt{control} \;\impl\; \non \EFs{\nb \texttt{cycle}
      < 10100} \texttt{alarm} \big).
    \]

  \item The warning and the alarm are maintained: 
    \[
    \AG (\texttt{warning} \; \impl \; \All \: \texttt{warning} \: \U
    \: (\texttt{alarm} \: \ou \: \texttt{control}) \big)
    \]
    and
    \[
    \AG (\texttt{alarm} \; \impl \; \All \: \texttt{alarm} \: \Uw \:
    \texttt{control} \big).
    \]
    Note that we use a \emph{weak} Until modality in the latter
    formula because we cannot ensure the occurrence of a control.
  \end{enumerate}

\item Consider a model for an ATM, whose atomic propositions include
  {\tt money}, {\tt reset} and {\tt error}, with the obvious meaning.
  To specify that it is not possible to get money after three mistakes
  were made in the same session (\ie with no intermediate reset), we
  can use the \CCTLba formula 
  \[ 
  {\AG \big( \non \EFs{\nb\texttt{error}\geq3 \et \nb\texttt{reset}=0}
    \texttt{money} \big)},
  \]
  or the \CCTLa formula
  \[
  \AG \big( \non \Ex (\non \texttt{reset}) \Us{\nb\texttt{error}\geq3}
  \texttt{money} \big).
  \]
\item Consider a mutual exclusion algorithm with $n$ processes trying
  to reach their critical section (CS). We can express a bounded
  waiting property with bound 10 (\ie\ when a process $P$ tries to
  reach its CS, then at most 10 other processes can reach theirs
  before $P$ does) by the $\CCTLba$ formula 
  \[
  \AG \ET_{i \in [1,n]} \big( \texttt{request}_i \impl \non
  \EFs{\sum_{j\not= i} \nb\texttt{CS}_{j} > 10 \et \nb\texttt{CS}_i=0}
  \top \big).
  \]
  As in the previous case, this can also be expressed in $\CCTLa$
  using $\U$ instead of $\F$.

\item In a model for a communicating system with events for the
  emission and reception of messages, the \CCTLpma formula $\AGs{\nb
    \texttt{send} - \nb \texttt{receive} < 0} \bot$ states that along
  any finite run, the number of {\tt receive} events cannot exceed the
  number of {\tt send} events.

\item Quantitative constraints can also be useful for fairness
  properties. For example the \CCTLba formula $\AG \: \AFs{\ET_i 5
    \leq \nb\vfi_i \leq 10} \top$ states that each $\vfi_i$ occurs
  infinitely often along every run (as does the \CTL formula $\ET_i
  (\AG \: \AF \: \vfi_i)$) but also ensures some constraint on the
  number of states satisfying formulas $\vfi_i$ along every execution:
  for example, it is not possible to have a sub-run where $\vfi_1$
  holds in 11 states and $\vfi_2$ in only 4 states.

\item Note that $\CCTLpm$ can express properties about the ratio
  between the number of occurrences of two kinds of states along a
  run. For example, $\EFs{100\cdot\nb\texttt{error} -\nb\top<0}P$ is
  true when there is a run leading to some state satisfying $P$ along
  which the rate of \texttt{error} states is less than 1
  percent. In fact any constraint of the form $\frac{\nb P}{\nb
    P'}\sim k$ can be expressed in this logic.

\item Finally note that we can use any temporal formula inside a
  constraint (and not only atomic propositions). For example, $\AG (
  \EFs{\nb(\EX \texttt{alarm})\leq 5} \texttt{init})$ states that it
  is always possible to reach \texttt{init} with a path along which at
  most 5 states have a successor satisfying \texttt{alarm}.
\end{enumerate}

\noindent Note that expressing these properties is rather
straightforward using counting constraints. When considering a
classical temporal logic, such properties cannot easily be expressed
directly. Unfolding the formula as it is done in the next section to
prove expressiveness results cannot be achieved in practice even when
the integer constraints are small: the formula would most of the time
become too long and too complex to be handled. A possible pragmatic
solution to avoiding counting constraints would be to add one or
several counters to the model and to use additional atomic
propositions to mark states (or rather, in such an extended model,
configurations) where the constraints over the values of counters are
satisfied. First note that this method may be less convenient or even
inapplicable in some cases, as it requires modifying the model under
verification. Moreover, this approach is difficult to use when
counting constraints do not only refer to atomic propositions, but
deal with nested temporal logic formulas (as in the last example
above) or even other counting properties, as this would require even
more drastic modifications to the model.

These examples illustrate the ability of our logics to state
properties over the portion of a run leading to some state. A similar
kind of properties could also be expressed with past-time modalities
(like $\S$ or $\Fm$), but unlike these modalities our constraints
cannot easily describe the ordering of events in the past: they
``only'' allow to count the number of occurrences of formulas. We will
see in the next sections that our extensions do not always induce a
complexity blow-up, while model-checking $\CTL+\Fm$ is known to be
\PSPACE-complete~\cite{laroussinie98}.

\section{Expressiveness and succinctness}
\label{sec:express}

When comparing two logics, the first question which comes to mind is
the range of properties they can be used to define, in other words
their \emph{expressiveness}. When they turn out to be equally
expressive, a natural way to distinguish them is then to ask \emph{how
  concisely} each logic can express a given property. This is referred
to as \emph{succinctness}, and is also relevant when studying the
complexity of model-checking for instance, since it may considerably
influence the size of a formula required to express a given property,
hence the time required to model-check it. 
In this section we study the expressiveness of the various logics
defined in the previous section, and provide results and comments
about their respective succinctness with respect to \CTL.

\subsection{Expressiveness}

We first show that only allowing Boolean combinations does not allow
our logics to express more properties than \CTL.

\begin{prop}
  \label{prop:csbb-to-ctl}
  Any $\CCTLb$ formula $\Phi$ can be translated into an equivalent
  \CTL formula of DAG-size $2^{O(|\Phi|^2)}$.
\end{prop}

\proof A naive translation, using nested $\Ex\_\U\_$ and $\All\_\U\_$
modalities to precisely count the number of times each subformula
inside a constraint is satisfied, is sufficient to show the
result. However the size of a translated formula would in general be
exponential in the \emph{value} of all integer constants and in the
DAG size of the original formula. We thus propose a more concise (yet
more involved) translation, whose size will be useful later on.

\medskip

Let $\Phi$ be a $\CCTLb$ formula. The proof is done by structural
induction over $\Phi$. The basic and Boolean cases are direct. By
Remark \ref{rem:until}, we only need to consider the cases $\Phi
= \EFs{C} \vfi$ and $\Phi = \AFs{C} \vfi$. Assume $C$ contains $m$
atomic constraints of the form $(\sum_{j \in [1,n_i]} \alpha^i_j
\nb\vfi^i_j) \sim k_i$ for $i \in [1,m]$. We translate $\Phi$ to $\CTL$
by building a family of formulas whose intended meaning is as follows:
\begin{iteMize}{$\bullet$}
\item If constraint $C$ holds with $\nb\vfi^i_j = 0$ for all $j,i$,
  then $\vfi$ may be true immediately.
\item Otherwise, successively check for every $j,i$ whether $\vfi^i_j$
  holds in the current state, and if so then update $C$ by decreasing
  the constant $k_i$ by $\alpha^i_j$.
\item Once all $\vfi^i_j$ have been scanned, proceed to the next state
  and re-evaluate $C$ for the new values of the constants.
\end{iteMize}
Let $\decr{C,i,j}$ denote the constraint obtained from $C$ by
replacing $k_i$ by $k_i - \alpha^i_j$. Note that in contrast with the
formal definition of \CCTLb constraints, we allow the $\mathit{decr}$
operation to result in negative constants in the right-hand sides of
atomic constraints. 

\medskip

Let $\bot$ and $\top$ be two special constraints satisfied by no
(resp. any) finite path in any Kripke structure, we also define the
constraint $\simp{C}$ obtained from $C$ by replacing any trivially
true atomic constraint (such as $S \geq 0$ or $S > -3$) by $\top$ and
any trivially false one (such as $S < 0$ or $S \leq -1$) by $\bot$,
and normalizing the obtained constraint in the usual way ($C\ou \bot
\rightarrow C$, \ldots).
%
% using the following rules :
% \begin{align*}
%   C \lor \bot & \rightarrow C & C \land \bot & \rightarrow \bot & C
%   \lor \top & \rightarrow \top & C \land \top & \rightarrow C
%   \\
%   \bot \lor C & \rightarrow C & \bot \land C & \rightarrow \bot & \top
%   \lor C & \rightarrow \top & \top \land C & \rightarrow C
%   \\
%   & & \lnot \top & \rightarrow \bot & \lnot \bot & \rightarrow \top
% \end{align*}
%
Note that due to this simplification step, $\simp{C}$ is either
reduced to $\top$ or $\bot$, or it does not contain $\top$ or $\bot$
as a sub-formula. Also note that $C$ and $\simp{C}$ are equivalent
(i.e. satisfied by the same finite runs).

\medskip

We now turn to the formal \CTL translation $\tr{\Phi}$ of formula $\Phi$,
which is defined inductively on the structure of $\Phi$. Boolean
combinations and negation are left unchanged. In the case where $\Phi
= \EFs{C} \vfi$, we proceed by unfolding the $\EF$ modality as follows:
\[
\tr{\EFs{C} \vfi} = \left\{
  \begin{array}{ll}
    \bot & \text{ if } \simp{C} = \bot
    \\
    \EF \tr{\vfi} & \text{ if } \simp{C} = \top
    \\
    \Ex \big(\ET_{i,j} \lnot \tr{\vfi^i_j}\big) \U \big(\tr{\vfi} \lor
    \Psi\big) & \text{ if } \epsilon \models C
    \\
    \Ex \big(\ET_{i,j} \lnot \tr{\vfi^i_j}\big) \U \Psi & \text{ if
    } \epsilon \not\models C
  \end{array}
\right.
\]
where $\Psi$ is a \CTL formula designed to be true in states where
both $\EFs{C} \phi$ and at least one formula $\vfi^i_j$ hold. Indeed
if $C$ is trivially false, then $\EFs{C} \vfi$ is clearly not
satisfiable. If $C$ is trivially true, it is sufficient to check that
$\vfi$ eventually holds without any further checks on sub-formulas
$\vfi^i_j$. The third case states that if $C$ holds on the empty path
then $\EFs{C} \vfi$ holds if, after a path prefix not affecting the
satisfaction of $C$, either $\vfi$ holds or some $\vfi^i_j$ holds and
we need to update $C$ again. The last case is identical except that it
does not check for $\vfi$ in the current state.
It then only remains to define $\Psi$. More generally, we describe a
family of $\CTL$ formulas $\Psi_{C,i,j,C'}$, where $i \in [1,m]$, $j
\in [1,n_i]$ with $m$ and $n_i$ as above, and $C, C'$ are $\CCTLb$
constraints. For all $1 \leq i \leq m, 1 \leq j \leq n_i$, let
\begin{equation}
  \Psi_{C,i,j,C'} = (\tr{\phi^i_j} \et \Psi_{C,i,j+1,\decr{C',i,j}})
  \ou (\non \tr{\phi^i_j} \et \Psi_{C,i,j+1,C'}).\label{eq:5}
\end{equation}
For all $1 \leq i < m$,
\begin{equation}
  \Psi_{C,i,n_i+1,C'} = \Psi_{C,i+1,1,C'}.\label{eq:3}
\end{equation}
Finally
\begin{equation}
  \Psi_{C,m,n_m+1,C'} = \left\{ 
    \begin{aligned}
      & \bot & \text{ if } C = C',
      \\
      & \EX \tr{\EFs{\simp{C'}} \vfi} & \text{ otherwise.}
    \end{aligned}
  \right. \label{eq:1}
\end{equation}
We then set $\Psi$ to denote $\Psi_{C,1,1,C}$. 
Formula $\Psi_{C,i,j,C'}$ implicitly assumes that, in the current
state, a certain (potentially empty) subset of the formulas $\vfi^1_1$
up to (but not including) $\vfi^i_j$ holds, and that $C'$ is the
constraint obtained by updating $C$ with respect to these
formulas. Then, it evaluates (the \CTL translation of) formula
$\vfi^i_j$, updating $C'$ if necessary and moving on to the next
sub-formula (Eq. \eqref{eq:5}). Whenever the scanning of sub-formulas
$\{\vfi^i_1, \ldots, \vfi^i_{n_i}\}$ relevant to the $i$-th atomic
constraint is finished, we proceed with the next one
(Eq. \eqref{eq:3}). Finally, once all sub-formulas have been scanned
and the constraint updated (Eq. \eqref{eq:1}), if no progress was made
at all (witnessed by the fact that $C = C'$), the formula is simply
deemed false. Otherwise we move to the next state along a possible run
using modality $\EX$, and develop the translation of formula
$\EFs{\simp{C'}} \vfi$ with the last updated $C'$.

\medskip

The above recursive definition characterizes finite formulas.  Indeed,
consider a formula $\tr{\EFs{\simp{C'}} \vfi}$ occurring as a sub-formula
of $\tr{\EFs{\simp{C}} \vfi}$, and $f$ the injection mapping each
right-hand-side constant $k'$ in $C'$ to the corresponding constant in
$C$. By definition, we have $f(k') \leq k'$ for every $k'$ in $C'$,
and either $f$ is not surjective (meaning that some constant $k$ in
$C$ no longer appears in $C'$ due to the simplification step in
Eq. \eqref{eq:1} above) or there exists $k'$ such that $f(k') <
k'$. This is guaranteed by the fact that developing $\Psi_{C,1,1,C}$
according to its definition into a formula containing
$\tr{\EFs{\simp{C'}} \vfi}$ resorts to at least one $\mathit{decr}$
operation followed by a simplification operation. Since any negative
constant appearing after a decrement is eliminated by the next
simplification step, this process cannot repeat indefinitely and must
therefore terminate.

The translation of $\AFs{C} \vfi$ is obtained by replacing each
occurrence of the path quantifier $\Ex$ by $\All$ in the above.  The
correctness of the translation can be shown by induction on the
nesting depth of until modalities in $\tr{\Phi}$ and quantities $m$
and $n_i$.

\medskip

We now turn to the worst-case DAG-size of the translation of the whole
\CCTL formula $\Phi$. Let $K$ be the largest integer constant in
$\Phi$, $M$ the maximal number of atomic constraints in any constraint
in $\Phi$ and $N$ the maximal number of counting expressions in any
atomic constraint in $\Phi$. The number of distinct $\Psi_{C,i,j,C'}$
formulas involved in the translation of any sub-formula $\EFs{C} \vfi$
or $\AFs{C} \vfi$ of $\Phi$ is bounded by $K^{M} \cdot M \cdot N \cdot
K^{M}$. This construction is repeated as many times as there are
temporal modalities in $\Phi$, which amounts to at most $|\Phi| \cdot
K^{M} \cdot M \cdot N \cdot K^{M}$ distinct sub-formulas (this
pessimistic upper bound clearly covers the case of Boolean
connectives, whose translation is much simpler). Since $M, N \in
O(|\Phi|)$ and $K \in O(2^{|\Phi|})$,
% and taking into account the translation of boolean connectives which
% accounts for $O(|\Phi|)$ additional sub-formulas,
we get a total DAG-size for $\tr{\Phi}$ in $O\big(|\Phi| \cdot
(2^{|\Phi|})^{|\Phi|} \cdot |\Phi| \cdot |\Phi| \cdot
(2^{|\Phi|})^{|\Phi|}\big) = O(|\Phi|^3.2^{2|\Phi|^2}) \subseteq
2^{O(|\Phi|^2)}$. \qed

\begin{exa}
  \label{exa:trans}
  For any integer $k$ and formula $\vfi$, we look at the translation
  of $\Phi_k = \EFs{C_k} \vfi$ where $C_k$ denotes the constraint
  $\nb p_1 + \nb p_2 = k$ and $\vfi$ is any formula:
  \begin{equation}
    \tr{\Phi_k} = \left\{
      \begin{aligned}
        & \textstyle \Ex \big(\ET_i \lnot p_i\big) \U \vfi & & \text{
          if } k = 0
        \\
        & \textstyle \Ex \big(\ET_i \lnot p_i\big) \U \Psi_k & &
        \text{ otherwise}
      \end{aligned}
    \right. \label{eq:trans1}
  \end{equation}
  \begin{multline*}
    \text{with \qquad} \Psi_k = \Big( p_1 \land \big( (p_2 \land \EX \tr{\simp{(\Phi_{k-2})}})
    \lor ( \lnot p_2 \land \EX \tr{\simp{(\Phi_{k-1})}}) \big) \Big)\\ \lor
    (\lnot p_1 \land p_2 \land \EX \tr{\simp{(\Phi_{k-1})}}),
  \end{multline*}
  where $\simp{(\Phi_{k})} = \Phi_{k}$ if $k \geq 0$ and $\bot$
  otherwise. Note that some simplifications were performed in this
  translation: namely, $(\vfi \lor \Psi_0)$ is replaced by $\vfi$ in
  the first case of Eq. \eqref{eq:trans1} since $\Psi_0 \equiv \bot$,
  and a conjunct containing $\bot$ is removed from $\Psi_k$.
\end{exa}

Note that we provided a parametric upper bound for the above
translation which can be interpreted for all variants of \CCTL below
$\CCTLb$. In contrast to this result, introducing subtractions in
constraints yields a strict increase in expressiveness.

\begin{prop}
  The $\CCTLpma$ formula $\varphi = \AGs{\nb A - \nb B < 0} \bot$
  cannot be translated into \CTL.
\end{prop}

\begin{proof}[(sketch)]
  Formula $\varphi$ (already seen in Sec. \ref{sec:examples-cctl} with
  different atomic propositions) states that the number of $B$-labeled
  states cannot exceed the number of $A$-labeled states along any
  path.
  As shown by \cite{bvw94} and also presented in \cite{wilke99}, the
  set of models of any \CTL formula can be recognized by a finite
  alternating tree automaton.
  Suppose there exists a \CTL formula $\varphi'$ equivalent to
  $\varphi$, and let $\mathcal{A}$ be the alternating tree automaton
  accepting its set of models. From $\mathcal{A}$, one can easily
  build a finite alternating automaton on words over $2^{\{A,B\}}$,
  whose accepted language is the set of all finite prefixes of
  branches in models of $\varphi$, namely words whose prefixes contain
  at most as many $B$'s as $A$'s. Since this language is clearly not
  regular, this leads to a contradiction.
\end{proof}

\subsection{Succinctness}

Our extensions of CTL come with three main potential sources of
concision, which appear to be orthogonal: the encoding of constants in
binary, the possibility to use Boolean combinations in constraints,
and the use of sums. However, only the first two turn out to yield an
exponential improvement in succinctness. First we consider the case of sums:

\begin{prop}
  \label{prop:cs1-to-ctl}
  For every formula $\Phi \in \CCTL$ with unary encoding of
  integers, there exists an equivalent \CTL formula of DAG-size
  polynomial in $|\Phi|$.
\end{prop}

\proof This proposition is a direct consequence of the DAG-size
computation presented in the proof of Prop. \ref{prop:csbb-to-ctl}
where $M$, the number of atomic constraints in a constraint in $\Phi$,
is set to 1 to reflect the absence of Boolean connectives inside
constraints, and where $K$, the maximal constant in $\Phi$, is bounded
by $|\Phi|$ due to the unary encoding. \qed

We now look at the succinctness gap due to the binary encoding of
constants\footnote{Note that for real-time logics, it is already known
  that the binary encoding of integer constants induces a complexity
  blow-up for the decision procedures~\cite{alur93b,alur94b}.}:

\begin{prop}
  $\CCTLa$ can be exponentially more succinct than $\CTL$.
\end{prop}

\begin{proof}
  In \cite{lst-TCS2001}, it is shown that the logic \TCTL, when
  interpreted over Kripke structures with a special atomic proposition
  \textsl{tick} used to mark the elapsing of time, can be
  exponentially more succinct than \CTL\footnote{This was also
    observed in~\cite{emerson92b} for the logic \RTCTL over
    $\DKSo$.}. More precisely, the \TCTL formulas $\EF_{<n} A$ and
  $\EF_{>n} A$, which are of size $O(\log(n))$ since $n$ is encoded in
  binary, do not admit any equivalent \CTL formula of temporal height
  (and hence also size) less than $n$. These formulas express the
  existence of a path where $A$ eventually holds and less (resp. more)
  than $n$ clock ticks are seen until then. They are respectively
  equivalent to the $O(\log(n))$-size $\CCTLa$ formulas $\EFs{\nb
    \textsl{tick} < n} A$ and $\EFs{\nb \textsl{tick} > n} A$.
\end{proof}

Note that the proof of the previous proposition only uses the simplest
kind of constraint: we do not need sums (and coefficients) or Boolean
combinations in the constraints.

This exhibits a first aspect in which \CCTL logics can be
exponentially more succinct than \CTL. However, as expressed in the
next proposition, another orthogonal feature of the logic may yield a
similar blow-up.

\begin{prop}
  $\CCTLba$ with unary encoding of integers can be exponentially
  more succinct than $\CTL$.
\end{prop}

\proof It was shown by \cite{wilke99,ai03} that any $\CTL$ formula
$\varphi$ equivalent to the $\CTL^+$ formula $\psi = \Ex(\F P_0 \land
\ldots \land \F P_n)$ must be of length exponential in $n$. It turns
out $\psi$ is equivalent to the $\CCTLba$ formula $\psi' =
\EFs{\bigwedge_i \nb P_i \geq 1} \top$, which entails the result. Note
that $\psi'$ only contains the constant 1, which means that this gap
cannot be imputed to the binary encoding. \qed

The intuitive reason for this blow-up is that a \CTL formula
expressing the property that atomic propositions $P_1$ to $P_n$ are
each seen at least once along a path would have to keep track of all
possible interleavings of occurrences of $P_i$'s.

To summarize, we showed that two different aspects of the extensions
of \CTL presented in this paper, while not increasing the overall
expressiveness of the logic, may yield exponential improvements in
succinctness.  It would remain to study the succinctness of remaining
$\CCTL$ fragments with respect to each other, in particular when these
aspects are combined.

\section{Model checking}
\label{sec:mc}

\subsection{Polynomial-time model-checking}

Even though, as we discussed in the previous section, diagonal
constraints lead to strictly more expressive logics than \CTL, it
turns out that model-checking $\CCTLpma$ is asymptotically not
more difficult than model checking \CTL itself. As a preliminary
result of independent interest, we show that the existence of a
polynomial-time algorithm for the model-checking of the logic \TCTL
over $\DKSzo$, as shown in \cite{lst-TCS2001}, remains true when
considering more general weighted graphs, namely DKS's with weights in
$\{-1,0,1\}$. This result will be used to establish the complexity of
model-checking for $\CCTLpma$, and as a corollary also for all
weaker fragments.

\begin{prop}
  \label{prop:tctl-ozo}
  The model-checking problem for $\TCTL$ over $\DKSozo$ is
  \P-complete.
\end{prop}

\P-hardness is inherited from \CTL (see \cite{phs-aiml02} for a proof
of the \P-hardness of \CTL). For membership in \P, we consider a DKS
$\calS = \tuple{Q,R,\ell}$ with $R\subseteq Q \times \{-1,0,1\} \times
Q$, a state $q \in Q$ and a \TCTL formula $\Phi$, and show that
deciding whether $q \models \Phi$ can be done in polynomial time. As
usual, we inductively assume the set of states satisfying all strict
sub-formulas of $\Phi$ to be known, and proceed from there. We
distinguish several cases:

\begin{enumerate}[(1)]
\item \label{case:eleq} $\Phi = \Ex \vfi \U_{\leq k} \psi$: We first
  determine the subset of states $\Quntil$ from which the $\CTL$
  formula $\Ex \vfi \U \psi$ holds, and consider the restriction
  $\calS'$ of $\calS$ to $\Quntil$ in which outgoing edges of states
  labeled by $\psi \et \non \vfi$ are removed. $\Phi$ holds over some
  state $q$ in $\calS$ if and only if $q \in \Quntil$ and there exists
  a path of weight at most $k$ in $\calS'$ from $q$ to some other
  state $q'$ where $\psi$ holds. Considered paths are either simple,
  or composed of a prefix from $q$ to some state $q''$, a
  negative-weight cycle from $q''$ to itself repeated a certain number
  of times, and a suffix from $q''$ to $q'$.

  Even though finding a \emph{simple} path of weight less than $k$ in
  a graph containing negative cycles is \NP-complete, this is not
  exactly what we are considering since we allow paths containing
  repeated states. Our problem can thus be tested in polynomial time
  using the classical Floyd-Warshall algorithm (to compute all-pairs
  shortest paths) over $\calS'$. The matrix $(\alpha)_{q,q'}$ of
  shortest-path weights computed by this algorithm gives us sufficient
  information: to decide whether a state $q$ satisfies $\Phi$: $q\sat
  \Phi$ one simply need to check whether there exists $q'$ satisfying
  $\psi$ such that either $\alpha_{q,q'} \leq k$, or there exists
  $q''$ such that $\alpha_{q,q''} < \infty$, $\alpha_{q'',q''} < 0$
  and $\alpha_{q'',q'} < \infty$.

\medskip
\item $\Phi = \Ex \vfi \U_{\geq k} \psi$: We build the DKS $\calS'$ as
  in the previous case, and a new DKS $\calS''$ isomorphic to $\calS'$
  but with opposite weights. Then, $\Phi$ is satisfied from $q$ in
  $\calS'$ (and thus also $\calS$) if and only if the formula $\Ex
  \vfi \U_{\leq -k} \psi$ is satisfied from $q$ in $\calS''$.

\medskip
\item $\Phi = \Ex \vfi \U_{= k} \psi$: We build the DKS $\calS'$ as in
  case \ref{case:eleq}, and compute the relation 
  \[
  R_k = \{ (q, q') \in \Q{\vfi \land \Ex \vfi \U \psi} \times \Quntil
  \mid q \Lra^{k}_R q' \}.
  \]
  For $k = 0$, $R_0$ can be seen as $\bigcup_{i \geq 0} X_i$ with:
  \[
  \left\{
  \begin{aligned}
    X_0 &\ = (\lra^{0}_R)^*
    \\
    X_{i+1} &\ = X_{i} \cup (X_i \cdot \lra^{1}_R \cdot X_i \cdot \lra^{-1}_R \cdot
    X_i) \cup (X_i \cdot \lra^{-1}_R \cdot X_i \cdot \lra^{1}_R \cdot X_i)
  \end{aligned}
  \right.
  \]
  which can be obtained by a simple fixed-point computation requiring
  at most $|Q|^2$ iterations (since $|R_0| \leq |Q|^2$). For $k = 1$,
  we simply have $R_1 = R_0 \cdot \lra^{1}_R \cdot R_0$.
  For greater values of $k$, we use dichotomy to express this relation
  in terms of $R_0$ and $R_1$ in $O(\log(k))$ steps (\ie $O(|\Phi|)$,
  since $k$ is encoded in binary), by writing
  \[
  R_k = R_{\floor{k/2}} \cdot R_{\ceil{k/2}}.
  \]
  Each of these relational compositions requires time at most cubic in
  the size of $Q$. It then suffices to test whether $(q, q') \in R_k$
  for some $q'$ verifying $\psi$.

\medskip
\item $\Phi = \All \vfi \U_{= 0} \psi$: The procedure consists in
  defining a standard Kripke structure $\calS'$ and a classical \CTL
  formula $\Psi$ such that $\calS'$ satisfies $\Psi$ if and only if
  $\calS$ does not satisfy $\Phi$.

  Using fixed-point computations over $Q \times Q$, we compute the
  relations $\Rp{0}$ and $\Rm{0}$ as the respective least solutions of
  \[
  \left\{
  \begin{aligned}
    X_0 &\ = (\lra^{0}_R)^*
    \\
    X_{i+1} &\ = X_{i} \cup (\lra^{1}_R \cdot X_i \cdot \lra^{-1}_R)
  \end{aligned}
  \right.
  \text{\quad and \quad}
  \left\{
  \begin{aligned}
    X_0 &\ = (\lra^{0}_R)^*
    \\
     X_{i+1}  &\ = X_{i} \cup (\lra^{-1}_R \cdot X_i \cdot \lra^{1}_R).
  \end{aligned}
  \right.
  \]
  $\Rp{0}$ and $\Rm{0}$ respectively express the reachability relation
  in $\calS$ along paths of weight $0$ with no prefix of strictly
  negative (resp. positive) weight.
  We also define the relation $R_0^s$ (where $s$ stands for
  \emph{strict}) as:
  \[
  R_0^s =\ \lfl{0} \cup\ (\lfl{1} \cdot \Rp{0} \cdot \lfl{-1}) \cup
  (\lfl{-1} \cdot \Rm{0} \cdot \lfl{1}),
  \] 
  which expresses reachability in $\calS$ by 0-weight paths
  such that no intermediate state (other
  than the initial one) is reached with weight 0.
  Let $Q^+, Q^-$ be two isomorphic copies of $Q$ (and $q^+$, $q^-$
  denote the copies in $Q^+$ and $Q^-$ of some state $q \in Q$), we
  can now construct $\calS' = (Q',R',\ell')$ with $Q' = Q \union Q^+
  \union Q^-$, $\ell'(q) = \ell(q)$ if $q \in Q$ and $\ell'(q^\pm) =
  \ell(q) \cup \{{\sf ok}\}$, and
  \begin{align*}
    R' =~ & \{q_1 \rightarrow q_2 \mid (q_1, q_2) \in R_0^s \}
    \\
    \cup\ & \{q_1 \rightarrow q_2^+ \mid q_1 \lra^1_R q_2\} \cup
    \{q_1^+ \rightarrow q_2^+ \mid q_1 \lra^1_R q_2 \ou (q_1, q_2) \in
    \Rp{0} \}
    \\
    \cup~ & \{q_1 \rightarrow q_2^- \mid q_1 \lra^{-1}_R q_2\} \cup
    \{q_1^- \rightarrow q_2^- \mid q_1 \lra^{-1}_R q_2 \ou (q_1, q_2)
    \in \Rm{0} \}.
  \end{align*}
  In order to eliminate finite paths, we additionally complete
  $\calS'$ with a dummy state $q_\bot$ and transitions from every
  state to $q_\bot$ and a loop from $q_\bot$ to itself. We let
  $\ell'(q_\bot) = \{\psi\}$, which will be explained in detail later
  on.

  The set of states of $\calS'$ is divided into four subsets: states
  in $Q$ correspond to the states reachable with weight 0 in $\calS$,
  and states in $Q^+$ and $Q^-$ are the states reachable with weight
  strictly more or strictly less than 0. Paths in $\calS'$ ending in
  the dummy state may not correspond to actual paths in $\calS$, but
  they correspond to situations which are irrelevant to solving the
  problem. Since a path going from $Q$ to $Q^+$, and then from $Q^+$
  to $Q$ is captured by the relation $R_0^s$, we can omit transitions
  going back to $Q$ from $Q^+$ (and similarly for $Q^-$). Hence all
  runs of $\calS'$ either stay forever in $Q$, eventually reach $Q^+$
  or $Q^-$ and stay there forever, or reach the dummy state and stay
  there forever.

  We now define the \CTL formula $\Psi$ as $\Ex (\lnot \psi \ou {\sf
    ok}) \Uw (\lnot \vfi \et (\lnot \psi \ou {\sf ok}))$ 
  \footnote{$\Uw$ is called the \emph{weak until} modality, and
    $\vfi \Uw \psi$ holds along a path if either $\G \vfi$ or $\vfi \U
    \psi$ does.} and claim
  that $q \models_{\calS'} \Psi$ if and only if $q \models_{\calS}
  \lnot \Phi$.  The idea of the proof is to show that if $\Phi$ is not
  satisfied from some state $q$ in $\calS$ then one can find a path
  from $q$ in $\calS'$ satisfying $\Psi$, and conversely that finding
  a path satisfying $\Psi$ from $q$ over $\calS'$ is sufficient to
  disprove $\Phi$ from that state in $\calS$.
  
  \begin{lem}
    $q \sat_\calS \lnot \Phi \implies q \sat_{\calS'} \Psi$.
  \end{lem}

  \proof There are several ways in which $\Phi$ may fail to hold over
  $\calS$:
  \begin{enumerate}[(a)]
  \item \label{e:nonphi} There exists a path $\rho$ in $\calS$ along
    which a state $q_1 \sat \lnot \vfi$ appears strictly before the first
    state satisfying $\psi$ and reached with weight 0. Let $\rho_1
    q_1$ be the shortest prefix of $\rho$ such that $q_1 \sat \lnot
    \vfi$ and either $q_1 \sat \lnot \psi$ or $\wght{\rho_1 q_1} \neq
    0$.
    \begin{enumerate}[(i)]
    \item If $q_1 \sat \lnot \psi$ and $\wght{\rho_1 q_1} = 0$, then
      by definition of $R_0^s$ there must exist a path $\rho_1'$ from
      $q$ to $q_1$ in $\calS'$ whose intermediate states all satisfy
      $\lnot \psi$. Consequently, any infinite continuation of $\rho'$
      must satisfy $\lnot \psi \Uw (\lnot \vfi \et \lnot \psi)$, which
      implies that $q \sat_{\calS'} \Psi$.
    \item \label{ee:pm} If $\wght{\rho_1
        q_1} \neq 0$, then we can write $\rho_1 q_1 = \rho_2 q_2
      \rho_3 q_1$ where $\rho_2 q_2$ is the longest prefix of $\rho_1$
      of weight 0. By definition, $q_2 \rho_3 q_1$ starts with a non-0
      transition and has no prefix of weight 0, hence by definition of
      $\calS'$ there must exist a finite path $\rho' = \rho'_2 q_2
      \rho'_3 q_1^\pm$ in $\calS'$ such that all intermediate states
      of $\rho'_2 q_2$ satisfy $\lnot \psi$ and all intermediate
      states of $\rho'_3 q_1^\pm$ satisfy ${\sf ok}$. Hence any
      continuation of $\rho'$ must satisfy $(\lnot \psi \lor {\sf ok})
      \Uw (\lnot \vfi \et {\sf ok})$, which implies that $q
      \sat_{\calS'} \Psi$.
    \end{enumerate}
  \item There exists a path $\rho$ in $\calS$ along which no state
    satisfying $\psi$ ever appears at the end of a prefix of weight
    0. We assume that $\vfi$ consistently holds along the path,
    otherwise it comes down to the previous case. There are again two
    cases to consider:
    \begin{enumerate}[(i)]
    \item If $\rho$ has infinitely many prefixes of weight $0$, then
      by definition of $R_0^s$ there must exist an infinite path
      $\rho'$ in $\calS'$ whose intermediate states never leave the
      set $Q$ and all satisfy $\lnot \psi$. Therefore $\rho'$
      satisfies $\G \lnot \psi$, which implies that $q \sat_{\calS'}
      \Psi$.
    \item If $\rho$ has finitely many prefixes of weight $0$, then
      using ideas similar to the above, one can decompose it as
      $\rho_1 \rho_2$, with $\rho_1$ its longest finite prefix of
      weight 0, and $\rho_2$ an infinite path with no prefix of weight
      0. This implies the existence of a corresponding path $\rho'$ in
      $\calS'$ with a finite prefix remaining in $Q$ whose states all
      satisfy $\lnot \psi$ and an infinite suffix remaining in $Q^+$
      or $Q^-$ whose states all satisfy ${\sf ok}$. Therefore $\rho'$
      must satisfy $\G \lnot \psi \lor {\sf ok}$, which implies that
      $q \sat_{\calS'} \Psi$. \qed
    \end{enumerate}
  \end{enumerate}

  \begin{lem}
    $q \sat_{\calS'} \Psi \implies q \sat_\calS \lnot \Phi$.
  \end{lem}

  \proof The proof is very similar to that of the previous lemma. Let
  us consider a path $\rho'$ in $\calS'$ satisfying $\Psi$. There are
  two main possibilities:
  \begin{enumerate}[(a)]
  \item The path $\rho'$ consistently satisfies $\lnot \psi \ou {\sf
      ok}$. We distinguish two cases.
    \begin{enumerate}[(i)]
    \item If $\rho'$ never leaves the set $Q$ (and thus consists only
      of edges representing the relation $R_0^s$), then there must
      exist a corresponding path $\rho$ in $\calS$ visiting at least
      the same states in the same order (since $R_0^s$ is a
      restriction of the reachability relation of $\calS$). Moreover,
      \emph{all} states reached with weight $0$ in $\rho$ must appear
      in $\rho'$ (by definition of $R_0^s$). Now whether or not the
      states in $\rho$ satisfy $\phi$, $\Phi$ cannot be satisfied in
      $\calS$ from $q$ since no state reached with weight $0$
      satisfies $\psi$ along $\rho$.
    \item If $\rho'$ eventually leaves the set $Q$, and since there
      are no transitions out of $Q^+$ and $Q^-$ except to the dummy
      state $q_\bot$ (which satisfies $\psi$ but not ${\sf ok}$), then
      necessarily $\rho'$ can be decomposed into $\rho'_1 q_1 \rho'_2$
      where $\rho'_1 q_1$ is a finite path in $Q$ necessarily
      satisfying $\G \lnot \psi$ and $\rho_2$ is an infinite path
      either in $Q^+$ or $Q^-$ necessarily satisfying $\G {\sf
        ok}$. As previously, this implies the existence of a
      corresponding path $\rho$ in $\calS$, where the part
      corresponding to $\rho'_1$ never visits a state satisfying
      $\psi$ with weight 0, and the part corresponding to $\rho_2$
      never reaches weight 0 again. Thus $\Phi$ cannot be satisfied
      from $q$ in $\calS$.
    \end{enumerate}
  \item The other possibility is that $\rho'$ can be written $\rho'_1
    q_1 \rho'_2$, where $\rho'_1$ satisfies $\G \lnot \psi \ou {\sf
      ok}$ and $q_1$ satisfies $\lnot \vfi \land (\lnot \psi \lor {\sf
      ok})$. Again there are two possible cases.
    \begin{enumerate}[(i)]
    \item If $q_1 \in Q$, then $\rho'_1$ only visits states satisfying
      $\lnot \psi$, and $q_1 \sat \lnot \vfi \land \lnot \psi$. As
      previously there must exist a corresponding path $\rho_1$ in
      $\calS$ visiting at least the same states in the same order. Now
      since by definition of $R_0^s$ all 0-weight prefixes of $\rho_1$
      end in states appearing in $\rho'$ and satisfying $\lnot \psi$,
      and since $q_1$ satisfies $\lnot \vfi \land \lnot \psi$, no
      continuation of $\rho_1 q_1$ in $\calS$ can satisfy $\vfi \U_{=
        0} \psi$.
    \item If $q_1 \not\in Q$, then necessarily $q_1 \in Q^+ \cup Q^-$
      (since $q_\bot \not \sat \lnot \psi \lor {\sf ok}$) and $q_1
      \sat \lnot \vfi \land {\sf ok}$. Consequently one can write
      $\rho'_1 = \rho'_2 q_2 \rho'_3$ such that $q_2 \in Q$, $\rho'_2$
      never leaves $Q$ and $\rho'_3$ never leaves either $Q^+$ or
      $Q^-$. Moreover, all states in $\rho'_2 q_2$ satisfy $\lnot
      \psi$. One can thus build in $\calS$ a finite path $\rho$ from
      $q$ to $q_1$ going through $q_2$, in which no state reached with
      weight 0 up to $q_2$ (and thus also up to $q_1$) satisfies
      $\psi$, and all states occurring after $q_2$ (in particular
      $q_1$) are reached with non-0 weight. Hence since $q_1 \sat
      \lnot \vfi$, this implies that no continuation of $\rho$ can
      satisfy $\vfi \U_{= 0} \psi$. \qed
    \end{enumerate}
  \end{enumerate}
  \newcounter{saveenumi}
  \setcounter{saveenumi}{\theenumi}
\end{enumerate}

\begin{enumerate}[(1)]
  \setcounter{enumi}{\value{saveenumi}}

\item $\Phi = \All \vfi \U_{= k} \psi$: This case is similar to the
  previous one with slight modifications of the construction. We first
  assume $k$ to be positive, otherwise we can replace $S$ by an
  identical structure in which all weights are inverted and solve the
  formula with parameter $-k$. We then inductively compute
  \[
  \Rm{k} = \Rm{\lfloor k/2 \rfloor} \cdot \Rm{\lceil k/2 \rceil}
  \text{\quad with \quad} \Rm{1} = \Rm{0} \cdot \lfl{1}
  \]
  and $\Rm{0}$ defined as previously. $\Rm{k}$ is the reachability
  relation in $G_{\calS}$ by paths of weight $k$ whose prefixes all
  have weight strictly less than $k$.  We also compute, using for
  instance a modified Floyd-Warshall algorithm in which all integers
  greater than or equal to $k$ are assimilated to $\infty$, the
  reachability relation $\Rm{<k} = \{ (q, q') \mid \exists \sigma = q
  \sigma' q', \forall \rho\leq\sigma,\ \wght{\rho} < k \}$.

  We now construct a Kripke structure $\calS'$ as in the previous
  case, except that $Q' = Q \union Q^{init} \union Q^+ \union Q^- \cup
  \{q_\bot\}$ (where $Q^{\mathit{init}}$ is yet another copy of $Q$
  and $q^{\mathit{init}}$ denotes the copy of $q$ in
  $Q^{\mathit{init}}$) and $R'$ also contains $\{q_1^{init}
  \rightarrow q_2 \mid q_1 \Rm{k} q_2\} \union \{q_1^{init}
  \rightarrow q_2^- \mid q_1 \Rm{<k} q_2\}$. We additionally label
  states in $Q^{\mathit{init}}$ with the atomic proposition ${\sf
    ok}$. With this new Kripke structure, we can show that $q
  \models_{\calS} \lnot \Phi$ if and only if $q^{init}
  \models_{\calS'} \Psi$.
  
\medskip
\item $\Phi = \All \vfi \U_{\sim k} \psi$ with $\sim \in \{\leq, <, >,
  \geq\}$: Let us first treat the case where $\sim$ is $\leq$. We
  assume $k$ to be greater than or equal to 0, otherwise we invert all
  weights in $\calS$ and solve the problem using the procedure for
  $\Phi = \All \vfi \U_{\geq -k} \psi$. We essentially use the same
  procedure as for the previous case ($= k$), with a few
  modifications:
  \begin{enumerate}[(a)]
  \item Relations $\Rm{x}$ have to be computed over the restricted set
    of states $Q' = \{q \in Q \mid q \sat \lnot \psi\}$, because we
    have to make sure that no ``hidden'' intermediate state reached
    after a path of weight less than $k$ satisfies $\psi$;
  \item States in $Q^-$ should no longer be labelled by atomic
    proposition ${\sf ok}$, because paths which ultimately remain in
    $Q^-$ may correspond to paths in $\calS$ satisfying $\Phi$, and
    thus should not satisfy $\Psi$ unlike previously;
  \item Similarly, we remove the label ${\sf ok}$ from states in
    $Q^{\mathit{init}}$, that is $\forall q \in Q,\
    \ell'(q^{init}) = \ell(q)$.
  \end{enumerate}
    
  \noindent
  In the case where $\sim$ is $\geq$, we simply need to re-label
  states in $Q^{\mathit{init}}$ and $Q^-$ with ${\sf ok}$, and remove
  ${\sf ok}$ from the labelling of $Q^+$. Cases where $\sim$ is $<$
  and $>$ are dealt with by adding the ${\sf ok}$ label on states in
  $Q$ in the constructions for $\leq$ and $\geq$.
\end{enumerate}

\noindent
This concludes the proof that deciding the satisfaction of a $\TCTL$
formula from a given state of a $\DKSozo$ is in \P. \qed

\begin{thm}
  The model-checking problem for $\CCTLpma$ is \P-complete.
  \label{thm:mc-CS2}
\end{thm}

\proof
  As usual, \P-hardness is inherited from \CTL. Membership in \P\ is
  done by reduction to \TCTL model-checking over $\DKSozo$.

  We provide polynomial-time procedures to deal with the sub-formulas
  $\Ex \vfi \Us{C} \psi$ and $\All \vfi \Us{C} \psi$ with $C =
  \sum_{i=1}^\ell \alpha_i\nb\vfi_i \sim k$ where $\alpha_i \in
  \{-1,1\}$ and $k \in \mathbb{Z}$.
  Consider a Kripke structure $\calS = (Q,R,\ell)$, and
  inductively assume that the truth values of $\vfi$, $\psi$ and
  $\vfi_i$ over each state of $\calS$ are known: these sub-formulas
  will be seen as atomic propositions in the following.

  To each state $q$ occurring along a path, we associate a cost
  $\csize{q}{C} = \sum \{ \alpha_i \mid q \models \vfi_i \}$, and
  note that the \emph{value} of $\csize{q}{C}$ is in $O(|C|)$. This
  cost is additively extended to paths in the usual way. Deciding the
  truth value of the path formula $\vfi \Us{C} \psi$ 
  then amounts to checking whether
  there exists a finite prefix $\rho' q$ of $\rho$ such that
  $\csize{\rho'}{C} \sim k$, $q \models \psi$ and $\forall i \leq
  |\rho'|, \rho'(i) \models \vfi$.

  Given the type of our counting constraints, each state contributes
  to the cost of a path by a certain positive or negative number whose
  absolute value is bounded by $d = \max (\sum \{ \alpha_i \mid
  \alpha_i = 1 \}, \sum \{ |\alpha_i| \mid \alpha_i = -1 \})$. The
  idea is to build a durational Kripke structure with weights in
  $\{-1,0,1\}$, by adding (at most $d+1$) copies of each state in the
  original Kripke structure.

  Formally, we build from $\calS$ a $\DKSozo$ $\calS' = (Q',R',\ell')$
  as follows: for each state $q \in Q$ with $|\csize{q}{C}| = n$, $Q'$
  contains $n+1$ additional states $q_0, \ldots, q_n$. $R'$ is then
  defined as $\{ q \lfl{0} q_0 \mid q \in Q \} \cup \{ q_n \lfl{0} q'
  \mid (q, q') \in R, n = |\csize{q}{C}| \} \cup \{ q_i \lfl{\delta_q}
  q_{i+1} \mid q \in Q, i < |\csize{q}{C}| \}$ with $\delta_q = 1$ if
  $\csize{q}{C} >0$ and $\delta_q = -1$ otherwise. Finally, we set
  $\ell'(q_i) = \emptyset$ for all $q_i \in Q' \setminus Q$ and
  $\ell'(q) = \ell(q) \cup \{{\sf ok}\}$ for all $q \in Q' \cap Q$,
  where ${\sf ok}$ is a new atomic predicate.

  To each path $\rho = q \sigma$ in $\calS$, we associate the path
  $\tilde{\rho} = q q_0 \ldots q_n \tilde{\sigma}$ in $\calS'$. It can
  now be shown that $\rho$ satisfies $\vfi \Us{C} \psi$ if and only if
  $\tilde{\rho}$ satisfies the \TCTL path formula $({\sf ok} \impl
  \vfi) \Us{\sim k} ({\sf ok} \wedge \psi)$, and consequently that
  some state $q$ satisfies $\All \vfi \Us{C} \psi$ (resp. $\Ex \vfi
  \Us{C} \psi$) in $\calS$ if and only if it satisfies $\All ({\sf ok}
  \impl \vfi) \Us{\sim k} ({\sf ok} \wedge \psi)$ (resp. $\Ex ({\sf
    ok} \impl \vfi) \Us{\sim k} ({\sf ok} \wedge \psi)$) in $\calS'$.

  Suppose $\rho \models_{\calS} \vfi \Us{C} \psi$. We reason by
  induction on the least integer $i$ such that $\rho_{|i-1}
  \models_{\calS} C$, $\rho(j) \models_{\calS} \vfi$ for all $j < i$
  and $\rho(i) \models_{\calS} \psi$. If $i = 1$, then $\rho(1)
  \models_{\calS} \psi$ and thus $\tilde{\rho}(1) \models_{\calS'}
  \psi$ (recall that $\psi$ is seen as atomic). Otherwise, $\rho = q
  \rho'$ with $q \models_{\calS} \vfi$ and $\rho' \models_{\calS} \vfi
  \Us{C'} \psi$ with $C' = \sum_{i=1}^\ell \nb\vfi_i \sim k -
  \csize{q}{C}$, in other words $\rho'_{|i-2} \models_{\calS} C'$ and
  $\rho'(i-1) \models_{\calS} \psi$. By induction hypothesis, we have
  $\tilde{\rho}' \models_{\calS'} ({\sf ok} \thn \vfi) \Us{\sim k -
    \csize{q}{C}} ({\sf ok} \wedge \psi)$. Hence $\tilde{\rho} = q q_0
  \ldots q_{\csize{q}{C}} \tilde{\rho}' \models_{\calS'} ({\sf ok}
  \thn \vfi) \Us{\sim k} ({\sf ok} \wedge \psi)$.

  Conversely, consider a path $\rho$ in $\calS'$ starting with some
  state $q \in Q$ such that $\rho \models_{\calS'} ({\sf ok} \thn
  \vfi) \Us{\sim k} ({\sf ok} \wedge \psi)$, and as previously let $i$
  be the least integer such that $|\rho_{|i-1}| \sim k$, $\rho(j)
  \models_{\calS'} ({\sf ok} \thn \vfi)$ for all $j < i$ and $\rho(i)
  \models_{\calS'} ({\sf ok} \land \psi)$. By construction of
  $\calS'$, there must exist a unique path $\sigma$ in $\calS$ such
  that $\tilde{\sigma} = \rho$. We show by induction on $i$ that
  $\sigma \models_{\calS} \vfi \Us{C} \psi$. If $i = 1$, then $\rho(1)
  \models_{\calS'} {\sf ok} \wedge \psi$, in which case $\sigma(1)
  \models_{\calS} \psi$ holds in $\calS$. Otherwise by construction of
  $\calS'$ there must exist $q' \in Q$ such that $\rho = q q_0 q_1
  \ldots q_n q' \rho'$ and $q' \rho' \models_{\calS'} ({\sf ok} \thn
  \vfi) \Us{\sim k - \csize{q}{C}} ({\sf ok} \wedge \psi)$. Let
  $\sigma = q q' \sigma'$, by induction hypothesis we have $q' \sigma'
  \models_{\calS} \vfi \Us{C'} \psi$ with $C' = \sum_{i=1}^\ell
  \nb\vfi_i \sim k - \csize{q}{C}$. Hence $\sigma \models_{\calS} \vfi
  \Us{C} \psi$.  \qed

This result implies the following corollary on the complexity of
model-checking for all fragments of intermediate expressiveness:

\begin{cor}
  The model-checking problem for \CCTLa is \P-complete.
\end{cor}

Note that this weaker fragment allows considerable
simplification of the proof presented above for $\CCTLpma$. 
% In particular, model-checking $\CCTL_{\CS2}$ or $\CCTL_{\CS0}$ does
% not require the introduction of multiple copies of each state when
% converting the KS to a DKS, since in that case $\csize{q}{C}$ is
% already in $\{-1,0,1\}$ (resp. $\{0,1\}$).
Moreover, model-checking $\CCTLa$ can be done using the TCTL
model-checking algorithm provided in \cite{lst-TCS2001} instead of the
more involved construction used for Prop. \ref{prop:tctl-ozo}.

\subsection{Model-checking \texorpdfstring{\CCTLba, \CCTL and
    \CCTLb}{CCTL and boolean CCTL}}

We now establish the complexity of model-checking for the fragments
$\CCTLba$, $\CCTL$ and $\CCTLb$ and show that these problems are all
\DDP-complete.  Let us first recall the definition of the complexity
class \DDP, one of the classes of the polynomial hierarchy.

\begin{defi}
  $\DDP = \P^\NP$ is the class of problems solvable in polynomial time
  with access to an oracle for some \NP-complete problem.
\end{defi}

\noindent
We now prove \DDP-hardness of the model-checking problem
for $\CCTLba$.

\begin{thm}
  \label{thm:mc-BCS0}
  The model-checking problem for $\CCTLba$ is \DDP-hard.
\end{thm}

\proof
  % Membership comes from Theorem~\ref{theo-mc-BCSA1}.
  % 
We proceed by reduction from the \DDP-complete problem SNSAT
(sequentially nested satisfiability of propositional
logic)~\cite{fossacs2001-LMS}.

  Given $p$ families of variables~$X_1, \ldots X_p$ with $X_i =
  \{x_i^1, ..., x_i^m\}$ and a set $Z = \{z_1,\ldots, z_p\}$ of $p$
  variables, an instance $\calI$ of SNSAT is defined as a collection
  of $p$ propositional formulas $\phi_1, \ldots, \phi_p$ under
  $3$-conjunctive normal form (3-CNF), where each $\phi_i$ involves
  variables in $X_i \cup\{z_1,...,z_{i-1}\}$, and the value of each $z_i$
  is defined as $z_i = \exists X_i.\ \phi_i(z_1,...,z_{i-1},X_i)$.
  The instance $\calI$ is positive iff the value of $z_p$ is $\top$.
  We denote by $v_\calI$ the unique valuation of variables in $Z$
  induced by $\calI$.
   
  From $\calI$, we define the Kripke structure described in
  Figure~\ref{fig:snsat}.  Every state $z_i$ or $x_i^j$ is labeled by
  its name, every state $\bar{z}_i$ is labeled by some
  new atomic proposition $\bar{z}$ and every state of the form $q_i$
  is labeled by $q$.  We use $X$ to denote the set
  $X_1 \cup \cdots \cup X_p$ and $\calV$ for $X \cup Z$.
  A path $\rho$ from $q_{p}$ to $q_F$ describes the valuation $v_\rho$
  such that $v_\rho(y)=\top$ if $\rho$ visits state $y$ and $\bot$
  if it visits $\bar{y}$ for every variable $y$ in $\calV$.  We
  use a $\CCTLba$ formula to ensure that $v_\rho$ coincides with
  $v_\calI$ over $Z$, that is: $v_\rho(z_i)=\top$ iff $v_\calI(z_i)=\top$ for
  any $i \in \{1, \ldots, p\}$.

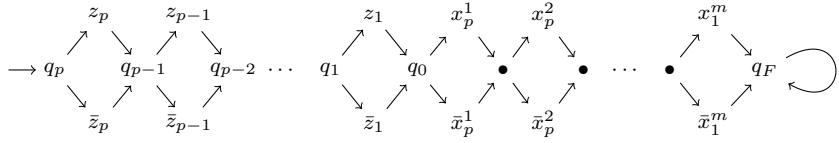
\begin{figure}[t]
  \centering
  \begin{tikzpicture}[every join/.style={->}]
    \scriptsize

     \matrix[row sep=2.5mm, column sep= 1mm] {
      & \node (zp) {$z_p$}; & & \node (zp-1) {\makebox[1em][c]{$z_{p-1}$}}; & & &
      & \node (z1) {$z_1$}; & & \node (xp1) {$x_p^1$}; &
      & \node (xp2) {$x_p^2$}; & & & & \node (x1m) {$x_1^m$}; \\

      \node[initial,initial text={}] (qp) {$q_{p}$}; & & \node (qp-1) {\makebox[1em][c]{$q_{p-1}$}}; 
      & & \node (qp-2) {\makebox[1em][c]{$q_{p-2}$}}; 
      & \node {$\cdots$}; &
      \node (q1) {$q_{1}$};  & & \node (q0) {$q_{0}$};  
      & & \node (q-1) {$\bullet$}; & & \node (q-2) {$\bullet$};
      & \node {$\cdots$}; &
      \node (q-3) {$\bullet$}; & & \node (qf) {$q_F$}; \\

      & \node (nzp) {$\bar{z}_p$}; & & \node (nzp-1) {\makebox[1em][c]{$\bar{z}_{p-1}$}}; & & &
      & \node (nz1) {$\bar{z}_1$}; & & \node (nxp1) {$\bar{x}_p^1$}; & 
      & \node (nxp2) {$\bar{x}_p^2$}; & & & & \node (nx1m) {$\bar{x}_1^m$}; \\
     };

     { [start chain] 
       \chainin (qp); \chainin (zp) [join]; \chainin (qp-1) [join];
       \chainin (zp-1) [join]; \chainin (qp-2) [join]; \chainin (q1);
       \chainin (z1) [join]; \chainin (q0) [join]; \chainin (xp1)
       [join]; \chainin (q-1) [join]; \chainin (xp2) [join]; \chainin
       (q-2) [join]; \chainin (q-3); \chainin (x1m) [join]; \chainin
       (qf) [join];
     }

     { [start chain] 
       \chainin (qp); \chainin (nzp) [join]; \chainin (qp-1) [join];
       \chainin (nzp-1) [join]; \chainin (qp-2) [join]; \chainin (q1);
       \chainin (nz1) [join]; \chainin (q0) [join]; \chainin (nxp1)
       [join]; \chainin (q-1) [join]; \chainin (nxp2) [join]; \chainin
       (q-2) [join]; \chainin (q-3); \chainin (nx1m) [join]; \chainin
       (qf) [join];
     }

     \path (qf) edge [->,in=-30,out=30,loop] (qf);

  \end{tikzpicture}
  \caption{Kripke structure associated to an SNSAT problem.}
  \label{fig:snsat}
\end{figure}

Let $\widetilde{\phi_i}$ be the formula $\phi_i$ where every
occurrence of the \emph{literal} $x$ is replaced by $\nb x \!=\!1$.
We define the $\CCTLba$ formula $\Psi_0$ as $\top$ and for every
$1\leq k \leq p$, $\Psi_k$ as $\EX \big(\Ex ({\bar{z}} \impl \non
\Psi_{k-1}) \Us{C_k} q_F\big)$, with $C_k \egdef \ET_{\ell\leq k}
\big((\nb z_\ell\!=\!1) \impl \widetilde{\phi_\ell}\big) \et \ET_{j=1}^{k}
\big( (\nb q \!=\!j) \impl \widetilde{\phi_j}\big)$.
The first part of the constraint $C_k$ aims at ensuring that
$v_\rho(z_\ell)=\top$ is witnessed by a valuation for
$\{z_1,\ldots,z_{\ell-1}\}\cup X_\ell$ satisfying $\phi_\ell$. The second part
ensures the formula $\phi_j$ is satisfied by $v_\rho$ when $\Psi_k$ is
interpreted from $z_j$ or $\bar{z}_j$ (\ie\ when the number of $q$'s
along the path leading to $q_F$ is $j$).
The formula $\Psi_j$ holds for a state $q_i$ with $i\leq j$ when
$v_{\calI}(z_i)$ is $\top$. The embedding of $\Psi_{j-1}$ inside
$\Psi_j$ is used to ensure that going through a $\bar{z}_m$ with $i
\geq m$ is always necessary w.r.t.\ $\calI$ (\ie\ there is no way to
satisfy the corresponding $\vfi_m$):

\begin{lem}
\label{lem-delta2}
For any $i=1,\ldots,p$ and  $i \leq j \leq p$, we have: 
$z_i \sat \Psi_j \equivaut v_{\calI}(z_i)=\top$ and 
$\bar{z}_i \not\sat \Psi_j \equivaut v_{\calI}(z_i)=\bot$
\end{lem}

\proof 
  First note that the truth value of $\Psi_j$ at $z_i$ and
  $\bar{z}_i$ is the same, due to the structure of paths and the fact
  that $\Psi_j$ begins with operator $\EX$. Therefore, both statements
  of the lemma are actually equivalent. Their proof is done by
  induction on $i$.
  
  \begin{iteMize}{$\bullet$}
  \item $i=1$: Any formula $\Psi_j$ with $ 1 \leq j \leq p$ holds
    from $z_1$ iff $q_0$ satisfies $\EFs{C_j} q_F$. And given the
    definition of $C_j$ and the structure of any path starting from
    $q_0$, this is equivalent to $q_0 \sat \EFs{\widetilde{\phi_1}}
    q_F$. And this last requirement is clearly equivalent to the
    existence of some valuation for $X^1$ to satisfy $\phi_1$.
    Finally note that $\bar{z}_1 \not\sat \Psi_j$ is equivalent to
    $q_0 \sat \non \EFs{\widetilde{\phi_1}} q_F$ and then
    $v_\calI(z_1)= \bot$.

  \item $i>1$: Knowing whether $z_i \sat \Psi_j$ is equivalent to
    $q_{i-1} \sat \Ex ({\bar{z}} \:\impl\: \non \Psi_{j-1})
    \Us{C'_j} \: q_F$ where $C'_j$ is the constraint $C'_j = \ET_{\ell
      \leq i} (\nb z_\ell\!=\!1 \impl \widetilde{\phi_\ell}) \: \et \:
    \widetilde{\phi_i}$.  This entails that there exists a path $\rho$
    leading to $q_F$ and defining a valuation $v_\rho$ such that:
    \begin{iteMize}{$-$}
    \item for any visited $z_\ell$ with $\ell<i$, we have $v_\rho \sat
      \phi_\ell$;
    \item for any visited $\bar{z}_\ell$ with $\ell<i$, ${\bar{z}}$ is
      true, and then $\non\Psi_{j-1}$ holds from $\bar{z}_\ell$. By
      induction hypothesis we have $v_\calI(z_\ell)=\bot$; and
    \item $v_\rho \sat \phi_i$.
    \end{iteMize}
    These three conditions define a valuation $v_\rho$ that coincides
    with $v_\calI$ for $\{z_{i-1}, \ldots, z_1\}$ and such that there
    exists a compatible valuation for satisfying $\phi_i$, thus
    $v_\calI(z_i)=\top$. Now if $\bar{z}_i \not\sat \Psi_j$, then
    $q_{i-1} \sat \non \Ex \big({\bar{z}} \:\impl\: \non \Psi_{j-1}
    \big) \Us{C'_j} \: q_F$ and then $v_\calI(z_i)\not=\top$. \qed
  \end{iteMize}

  It is now sufficient to check whether $q_0$ satisfies $\Psi_p$ or
  not, and then deduce the truth value of $v_\calI(z_p)$. 
\qed

Note that in the previous proof, one does not use sums in the
constraints to get the complexity lower bound.

\begin{thm}
  \label{thm:mc-CSA1}
  The model-checking problem for $\CCTL$ is \DDP-hard.
\end{thm}

\proof
  We provide a reduction from the model checking problem for \TCTL
  specifications over Durational Kripke structures. \TCTL formulas
  allow to deal with the cost (or duration) of paths (\ie\ the sum of
  the weight of every transition occurring along the path). This
  problem is \DDP-complete~\cite{LMS-tcs05}.
  Let $\calS = (Q,R_{\calS},\ell)$ be a DKS. Let $W$ be
  the set of weights occurring in $\calS$. We define the Kripke
  structure $\calS' = (Q',R_{\calS'},\ell')$ as follows:
  \begin{iteMize}{$\bullet$}
  \item $Q' \egdef Q \cup \{ (q,d,q')  \mid  \exists
    (q,d,q')\in R_\calS \}$,
  \item for any $(q,d,q') \in R_\calS$, we add $(q,(q,d,q'))$
    and $((q,d,q'),q')$ in $R_{\calS'}$; and
  \item $\ell':Q' \fleche 2^{\AP'}$ with $\AP' \egdef \AP \cup \{{\sf
      ok}\} \cup \{P_d \mid d \in W\}$, assuming ${\sf ok},P_d
    \not\in \AP$.  And we have: $\ell'(q) \egdef
    \ell(q)\cup\{\textsf{ok}\}$ for any $q\in Q$, and $\ell'(q,d,q') =
    \{P_d\}$.
  \end{iteMize}
  We also inductively define $\widetilde{\Phi}$ for any $\TCTL$ formula
  $\Phi$ as: $\widetilde{P} \egdef P$, $\widetilde{\non \psi} \egdef \non
  \widetilde{\psi}$, $\widetilde{\vfi \et \psi} \egdef
  \widetilde{\vfi} \et \widetilde{\psi}$, $\widetilde{\Ex \vfi
    \U_{\sim c} \psi} \egdef \Ex ({\sf ok} \impl \widetilde{\vfi})
  \Us{C_{\sim c}} ({\sf ok} \et \widetilde{\psi})$ and $\widetilde{\All \vfi
    \U_{\sim c} \psi} \egdef \All ({\sf ok} \impl \widetilde{\vfi})
  \Us{C_{\sim c}} ({\sf ok} \et \widetilde{\psi})$ with $C_{\sim c} \egdef
  \sum_{d\in W} d\cdot \nb P_d \sim c$.
  
  Now we can easily see that $q \sat_\calS \Phi$ with $\Phi \in \TCTL$
  is equivalent to $q \sat_{\calS'} \widetilde{\Phi}$.
   \qed

\begin{thm}
\label{thm:mc-BCSA1}
The model-checking problem for $\CCTLb$ is in \DDP.
\end{thm}

\proof
  Let $\calS = \tuple{Q, R,\ell}$ be a Kripke structure.  For this proof, 
  by definition of \DDP, it is sufficient to provide \NP\ procedures
  to deal with sub-formulas of the form $\EFs{C} \vfi$ and $\EGs{C}
  \vfi$ (Cf. Rem.\ref{rem:until}).
  First let $\{C_1, \ldots , C_m\}$ be the set of atomic constraints
  occurring in $C$.  Each $C_i$ is of the form $\sum_{j\in [1,n_i]}
  \alpha_{j}^i \cdot \nb\vfi_j^i \sim_i k_i$. And let $k$ be
  the maximal integer constant occurring in $C$. We can now present
  the algorithms:

  \begin{iteMize}{$\bullet$}
  \item $\Phi \egdef \EFs{C} \psi$: If $q \sat \Phi$, then there
    exists a run $\rho q'$ starting from $q$ such that $q'\sat\psi$
    and $\rho \models C$. First note that we can assume that the
    length of $\rho$ is bounded with respect to the model and formula
    (more specifically by $m\cdot\ |Q| \cdot (k+1)$): a sequence of
    $|Q|$ states contributes for at least $1$ to some linear
    expressions in $C$ (loops containing only $0$-states can be
    avoided since they do not contribute to the satisfaction of $C$)
    and every atomic constraint in $C$ needs at most to collect a
    total weight of $k+1$. Hence the length of $\rho$ is in
    $O(|Q|.2^{|C|})$ due to the binary encoding of the constants.

    An easy \NP\ algorithm consists in guessing the Parikh
    image\footnote{Recall that the Parikh image of a sequence $u$ over
      some alphabet $A$ is the function mapping each symbol in $A$ to
      its number of occurrences in $u$. This is also equivalently seen
      as a vector of dimension $|A|$ called the Parikh vector of $u$.}
    $F_\rho: R \fleche \Nat$ of the sequence of transitions in $\rho$,
    where $F_\rho(r)$ with $r \in R$ is the number of occurrences of
    transition $r$ in $\rho$.  As the length of $\rho$ is bounded by
    $m\cdot\ |Q| \cdot (k+1)$, $F_\rho$ can be represented in
    polynomial size.  Moreover one can check in polynomial time that:
    \begin{iteMize}{$\bullet$}
    \item $q'$ satisfies $\psi$,
    \item $\rho$ satisfies $C$, since $\nb\vfi_j^i = \sum_{r\ \mid\
        r\sat \vfi_j^i} \sum_{(r,r')\in R} F_\rho(r,r')$.
    \item $F_\rho$ corresponds to a correct path in $\calS$ (by
      verifying that the sub-graph induced by $F_\rho$ is connected
      and then applying the Euler circuit theorem).
    \end{iteMize}
    
  \item $\Phi \egdef \EGs{C} \psi$: For this case we have to find an
    infinite path $\rho$ satisfying the property ``whenever the current
    prefix satisfies $C$ then the next state has to satisfy $\psi$''.

    Every atomic constraint $C_i$ in $C$ may change its truth value at
    most twice along $\rho$. Therefore $\rho$ can be decomposed in at
    most $3m$ parts $(\rho_j)_{j \in [1,3m]}$ along each of which the
    truth value of every $C_i$ is constant. Of course a part can be
    empty (restricted to a single state) and the last part must
    contain a cycle to ensure that $\rho$ is infinite.

    As previously, the length of every $\rho_j$ is bounded and its
    Parikh image can be encoded in polynomial size. Moreover it is
    possible to ensure that each $\rho_j$ ends at the starting state
    of $\rho_{j+1}$. Finally we can also compute the truth value of
    $C$ over any sequence $\rho_1 \ldots \rho_j$ and then verify
    whether $\psi$ holds for any state in such a sequence if
    necessary. \qed
  \end{iteMize}

\noindent
A direct corollary of Theorems~\ref{thm:mc-BCS0}, \ref{thm:mc-CSA1}
and \ref{thm:mc-BCSA1} is:

\begin{cor}
  \label{coro-dd2}
  The model-checking problem for \CCTL, \CCTLba, \CCTLb is
  \DDP-complete.
\end{cor}

\subsection{Undecidability}

\begin{thm}
  The model-checking problem for $\CCTLbpma$ is undecidable.
  \label{thm:mc-BCS2}
\end{thm}

\proof
  This is done by reduction from the halting problem of a two-counter
  machine $\calM$ with counters $C$ and $D$, and $n$ instructions
  $I_1,\ldots,I_n$. Each $I_i$ is either a decrement
  $\tuple{\texttt{if}\ X \texttt{=0 then }\ j \ \texttt{else}$
    $X\texttt{--,}\ k}$ where $X$ stands for $C$ or
  $D$, an increment $\tuple{X\texttt{++,}\ j}$, or the
  halting instruction $\tuple{\texttt{halt}}$.
  We define a Kripke structure $\calS_\calM = (Q,R,\ell)$, where $Q
  = \{q_1, \ldots, q_n\} \cup \{r_i, s_i, t_i \mid I_i
  = \tuple{\texttt{if \ldots}}\}$. The transition relation is
  defined as follows:
  \begin{iteMize}{$\bullet$}
  \item if $I_i = \tuple{X \texttt{++,}\ j}$, then
    $(q_i,q_j) \in R$ ; and
  \item if $I_i = \tuple{\texttt{if}\ X\texttt{=0 then}\ j \
      \texttt{else}\ X \texttt{--,} k}$, then $(q_i,r_i)$,
    $(r_i,q_k)$, $(q_i,s_i)$, $(s_i,t_i)$, and $(t_i,q_j)$ in $R$.
  \end{iteMize}

  \noindent The labeling $\ell$ is defined over the set $\{\halt,
  \Cplus, \Cmoins, \Czero$, $\Czerob$, $\Dplus, \Dmoins, \Dzero$,
  $\Dzerob\}$ as $\ell(q_i) = \{\Xplus\}$ if $I_i$ is an increment of
  $X$, $\ell(r_i) = \{\Xmoins\}$, $\ell(s_i) = \{\Xzerob\}$ and
  $\ell(t_i) = \{\Xzero\}$ if $I_i$ is a decrement for $X$, and
  $\ell(q_i) = \{\halt\}$ if $I_i$ is the halting instruction.

  A run going through $s_i$ and $t_i$ for some $i$ will simulate the
  positive test ``$X=0$'': we use the propositions $\Xzerob$ and
  $X^0$ to observe this fact. Indeed along any run in $\calS_\calM$, a
  prefix satisfies $\nb \Xzerob > \nb \Xzero$ if and only if that prefix
  ends in some state $s_i$, which witnesses the fact that the
  counter's value was deemed equal to zero. The propositions on the
  other states are self-explanatory, witnessing increments and
  decrements of counters.

  Checking $\CCTLbpma$ on this structure solves the halting
  problem, since $\calM$ does not halt \emph{if and only if} $q_1
  \models_{\calS_{\calM}} \EGs{C} \bot$ with the following constraint:
  \[
  C \egdef (\nb \halt \geq 1)\ \ou \OU_{X\in\{\texttt{C},\texttt{D}\}}
  \!\!\Big( (\nb \Xplus - \nb \Xmoins < 0)\ \ou\ ( \nb \Xplus - \nb \Xmoins > 0 \et
  \nb \Xzerob-\nb \Xzero>0) \Big)
  \]
  This formula states that there exists a run where $C$ is
  consistently false, where $C$ is true either if the run terminates,
  or if the simulation of $\calM$ is wrong because the number of
  decrements is at some point larger than the number of increments, or
  because some counter was incorrectly assumed to be zero while
  simulating a test.
\qed

\section{Satisfiability}
\label{sec:satisf}

Here we address the satisfiability problem: given a formula $\Phi$,
does there exist a Kripke structure $\calS = \tuple{Q,R,\ell}$ with a
state $q\in Q$ such that $q \sat \Phi$?

For branching-time temporal logics, satisfiability problems are often
harder than model checking (contrary to linear-time temporal
logics)~\cite{emerson90}, this is also the case for our counting
logics. As soon as diagonal constraints are allowed (as in \CCTLpma or
\CCTLpm), satisfiability is undecidable: this can be easily shown by
adapting the undecidability proof of $\CCTLbpma$ model checking:

\begin{thm}
  The satisfiability problem for $\CCTLpma$ is undecidable.
  \label{satisf-undec}
\end{thm}

\proof As in the proof of Theorem~\ref{thm:mc-BCS2}, consider a
two-counter machine $\calM$ with counters $C$ and $D$, and $n$
instructions $I_1,\ldots,I_n$. We build a $\CCTLpma$ formula
$\Phi_\calM$ that is satisfiable iff $\calM$ halts.

We use the following set of atomic propositions: $\AP = \{ q_1,\ldots,
q_n, \Cplus, \Cmoins, \Czero, \Dplus, \Dmoins,$ $\Dzero, \halt\}$. The
$\CCTLpma$ formula $\Phi_\calM$ describes a linear KS whose every
state is labeled by exactly one $q_i$ corresponding to the current
state of $\calM$ and one proposition in $\calP = \{\Cplus, \Cmoins,
\Czero, \Dplus, \Dmoins, \Dzero, \halt\}$ that indicates the operation
that has to be done ($\Xplus$ and $\Xmoins$ are used to mark increment
and decrement of $X$, and $\Xzero$ labels states corresponding to an
instruction ``if $X==0$ \ldots'' when the current value of $X$ is
$0$). In the following we use $\calI_\calM(X)$ (resp. \
$\calT_\calM(X)$) to denote the set of instruction numbers
corresponding to an increment (resp.\ a test) of counter $X$.
$\Phi_\calM$ is the conjunction of the following formulae:

\begin{enumerate}[(1)]
\item ${\displaystyle \AG \Big( \OU_{i=1\ldots n} (q_i \et
    \ET_{j\not=i} \non q_j) \Big)}$
\item ${\displaystyle \AG \Big( \OU_{p\in \calP} (p \et \ET_{p' \in
      \calP\backslash \{p\}} \non p') \Big)}$
\item for every instruction, we have a step formula $\Phi_i$:
 
\begin{equation*}
\Phi_i  =  
\begin{cases} 
  \AG \big( q_i \:\impl\: (\Xplus \et \AX \: q_j) \big)
  & \text{if}\: I_i = \tuple{X \texttt{++,}\ j} \\
  \AG \Big( q_i \impl \big( (X_0 \et \AX \: q_j) \ou (\Xmoins \et \AX
  \: q_k) \big) \Big) & \text{if} \: I_i = \tuple{\texttt{if}\
    X\texttt{=0 then}\ j \
    \texttt{else}\ X \texttt{--,} k} \\
  \AG \big( q_i \impl (\halt \et \AX \: \halt) \big) & \text{if} \: I_i = \tuple{\texttt{halt}} \\
\end{cases}
\end{equation*}

\item no zero test succeeds when the actual value of the corresponding
  counter is strictly positive (i.e.  after a prefix witnessing
  strictly more increments than decrements), and no decrement is
  performed when that value is 0:
  \[
  \ET_{X\in\{C,D\}} \Big( \AGs{\nb \Xplus - \nb \Xmoins>0} (\non
  \Xzero) \;\et\; \AGs{\nb \Xplus =\nb \Xmoins} (\non \Xmoins) \Big)
  \]
%\;\et\; \AGs{\nb X_- - \nb X_+ >0} \bot

\item $\AF \: \halt$
\end{enumerate}

Clearly $\Phi$ is satisfiable by a finite KS iff $\calM$ terminates.\smallskip
\qed

\noindent For logics with no diagonal constraints, satisfiability remains
decidable, with an additional cost compared to classical \CTL.

\begin{thm}
  The satisfiability problems for logics ranging from $\CCTLa$ to
  $\CCTLb$ are 2-EXPTIME-complete .
  \label{satisf-2exp}
\end{thm}

\proof Hardness comes from the complexity of $RTCTL^{=}$
satisfiability~\cite{emerson92b}: this logic is an extension of $\CTL$
with an Until operator equipped with constraints of the form "$=\!k$''
over the number of transitions leading to the state satisfying the
right part of the Until. This result is based on an encoding of an
exponential space alternating Turing machine by a $RTCTL^{=}$
formula. Clearly, $RTCTL^{=}$ is included in $\CCTLa$.

2-EXPTIME membership directly follows from the translation given in
Lemma~\ref{prop:csbb-to-ctl}: any $\CCTLb$ formula can be translated
into \CTL and the resulting formula's DAG-size is in
$O(2^{|\Phi|^2})$. It remains to use an exponential algorithm for
$\CTL$ satisfiability to obtain a 2-EXPTIME procedure (note that
considering DAG-size instead of standard size does not matter for the
complexity of the $\CTL$ procedure: indeed, in~\cite{kupferman2000b}
for instance, the size of the alternating tree automaton built from a
given \CTL formula is its number of distinct subformulae).  \qed

\section{Extensions}
\label{sec:exten}

In the semantics of \CCTL modalities, each new path quantifier resets
the counting along a run, or more precisely starts counting anew on
the remaining portion of the run. This restriction is quite
significant, and ensures in particular that \CCTL is a
\emph{state-based} temporal logic. Under some circumstances (as well
as for the sake of completeness), it could be useful to relax this
hypothesis and consider logics in which nested modalities do not
necessarily reset the counting process.

In this section, we define two logics that allow this behaviour. The
first one, called $\CCTLv$, uses explicit variables to keep track of
the number of times a sub-formula was made true along the current run
since the variable was bound. The second logic, called $\CCTLc$ uses a
special reset modality and a different, cumulative semantics for
$\Us{C}$, where counting ranges over the whole portion of the run
since the last reset (hence potentially since the very beginning of
the run). This logic is interpreted over states with a history.

\subsection{Explicit variables}
\label{sec:explicit-variables}

Instead of using counting constraints associated with temporal
modalities, we now consider a logic equipped with explicit
\emph{variables} and constraints directly stated inside formulas.

 \begin{defi}
    \label{def-cctlv}
    Given a set of atomic propositions $\AP$ and a countable set of
    variables $V$, we denote by $\CCTLv$ the set of formulas of the
    form
    \[
    \vfi,\psi \grameg P \gramou \vfi \ou \psi \gramou \non \vfi
    \gramou \Ex \vfi \U \psi \gramou \All \vfi \U \psi \gramou
    z[\psi].\vfi \gramou \sum_{i=1}^\ell \alpha_i \cdot z_i \sim c
    \]
    where $P \in \AP$, $z, z_i \in V$, $\ell, \alpha_i, c \in \Nat$
    and $\sim\, \in\{<,\leq,=,\geq,>\}$.
\end{defi}

Intuitively $z[\psi].\vfi$ means that variable $z$ is defined and may
be used in formula $\vfi$, where it will stand for the number of times
formula $\psi$ was observed to be true along the current run since $z$
was defined.

More precisely, when the above formula is evaluated in a certain
state, (1) variable $z$ is reset to zero and bound to the sub-formula
$\psi$, (2) at each subsequent step of a run, $z$ is assigned the
number of states in which formula $\psi$ has held along this run since
$z$ was bound (\ie\ the value of $z$ evolves like $\nb\psi$ as in
Definition \ref{def:cctl-sem}) and (3) given this semantics for $z$,
$\vfi$ holds in the current state.

\begin{rem}
\label{rem:cctl-to-cctlv}
  The logic $\CCTLv$ can easily express any $\CCTLb$ property. Indeed,
  any $\CCTLb$ formula $\Ex \vfi \Us{C} \psi$, where $C$ is a boolean
  combination $f(C_1, \ldots, C_m)$ of atomic constraints $C_i =
  \sum_{j = 1}^{n_i} \nb \vfi^i_j \sim k_i$, is equivalent to the
  $\CCTLv$ formula
  \[
  z^1_1[\vfi^1_1] . z^1_2[\vfi^1_2] \ldots z^m_{n_m}
  [\vfi^m_{n_m}] . \Ex \vfi \U \big( \psi \land f(C'_1, \ldots, C'_m)
  \big)
  \]
  where $C'_i = \sum_{j = 1}^{n_i} z^i_j \sim k_i$ (and similarly for
  the $\All$-quantified modality). This translation yields formulas
  whose size is linear in that of the original formulas.

  For example, the $\CCTLba$ formula $\EFs{\nb P \leq 5 \et \nb P' >2}
  P''$, stating that there exists a run along which a state satisfying
  $P''$ is reached after at most 5 occurrences of $P$ and more than 2
  occurrences of $P'$, can be expressed in $\CCTLv$ as
  $z^1_1[P]. z^1_2[P']. \EF (z^1_1 \leq 5 \et z^1_2 > 2 \et P'' )$.
\end{rem}

We first introduce some notations. Given a function $f: E \rightarrow
F$, we denote by $\dom(f) \subseteq E$ the domain of $f$, and by
$\ran(f) \subseteq F$ its range. For $x\in E$ and $a \in F$, let $f[x
\leftarrow a]$ be the function mapping $x$ to $a$ and every $y \in
\dom(f) \setminus \{x\}$ to $f(y)$, and $f{|_{D}}$ be the restriction
of $f$ to some subset $D$ of $E$. 
Moreover we let $\subf(\vfi)$ be the set of all sub-formulas of $\vfi$
and $\V(\vfi)$ denote the set of all variables occurring in $\vfi$. An
occurrence of some $z \in \V(\Psi)$ is \emph{bound} if it occurs in
the right-hand side $\vfi$ of some sub-formula $z[\psi].\vfi \in
\subf(\Psi)$, and \emph{free} otherwise. A variable is free in $\Psi$
if it has at least one free occurrence. A formula without any free
variable is called \emph{closed}. Formally, the set $\FV(\Psi)
\subseteq \V(\Psi)$ of free variables of $\Psi$ is
\begin{align*}
  &\,\FV(\vfi_1 \ou \vfi_2) = \FV(\Ex \vfi_1 \U \vfi_2) = \FV(\All
  \vfi_1 \U \vfi_2) = \FV(\vfi_1) \union \FV(\vfi_2)
  \\
  &\begin{aligned} &\FV(P) = \emptyset & \qquad & \textstyle
    \FV(\sum_{i=1}^\ell \alpha_i \cdot z_i \sim c) = \{z_i \mid i \in
    [1,\ell]\}
    \\
    &\FV(\non \vfi) = \FV(\vfi) & & \FV(z[\psi].\vfi) = \FV(\psi) \cup
    (\FV(\vfi) \setminus \{z\})\end{aligned}
\end{align*}

\begin{rem}
  \label{rem:cctlv-restriction}
  In order to define the formal semantics of $\CCTLv$, one must be
  able to determine, in a given context, which sub-formula $\psi$ is
  bound to each variable $z$. For simplicity, we will henceforth make
  the following two assumptions on the syntax of formulas:
  \begin{enumerate}[(1)]
  \item In any formula, every variable is bound at most once. In other
    words, every subformula $z[\psi].\vfi$ deals with a distinct
    variable $z$. \label{item:syntax1}
  \item In any formula $\Phi$, there exists a (strict) total ordering
    $\prec$ on $\V(\Phi)$ such that any formula bound to some variable
    $z$ only contains occurrences of variables less than $z$, or more
    formally, for any sub-formula $z[\psi].\vfi$ of $\Phi$, $z' \in
    \V(\psi)$ implies $z' \prec z$. \label{item:syntax2}
  \end{enumerate}
  Note that neither assumption restricts the expressiveness of the
  logic, since one may easily rename variable occurrences in any
  formula to fulfill constraint \ref{item:syntax1}, and order
  variables according to an infix traversal of a formula's syntax tree
  to fulfill constraint \ref{item:syntax2}.
\end{rem}

We call \emph{environment} any partial function $\varepsilon: V
\rightarrow \CCTLv$. A pair $(\Phi,\varepsilon)$ where $\Phi$ is a
$\CCTLv$ formula and $\varepsilon$ is an environment, is a called a
\emph{closure}. We distinguish a specific class of closures, called
\emph{consistent}, defined as follows:

\begin{defi}
  \label{def:consistent}
  A \CCTLv closure $(\Phi,\varepsilon)$ is said to be
  \emph{consistent} if
  \begin{enumerate}[(1)]
  \item $\dom(\varepsilon) \cap \V(\Phi) = \FV(\Phi)$;
  \item for all $z \in \dom(\varepsilon)$ and $z' \in
    \FV(\varepsilon(z))$, $z' \in \dom(\varepsilon)$;
  \item for all $z \in \dom(\varepsilon)$ and $z' \in
    \V(\varepsilon(z))$, $z' \prec z$.
  \end{enumerate}
\end{defi}

Condition (1) guarantees that the environment for $\Phi$ defines at
least all free variables in $\Phi$ (and potentially some additional
variables not occurring in $\Phi$) and does not redefine any of
$\Phi$'s variables, condition (2) that $\varepsilon$ does not refer to
undefined variables and condition (3) that there are no cyclic
definitions. Note that for any closed formula $\Phi$, $(\Phi,
\varepsilon_{\emptyset})$ is consistent, where $\varepsilon_\emptyset$
is the empty environment.

A consistent $\CCTLv$ closure $(\vfi,\varepsilon)$ is interpreted over
a state of a Kripke structure extended with a valuation $v: V \fleche
\Nat$ such that $\dom(v) = \dom(\varepsilon)$.
% New version
Given a consistent closure $(\vfi,\varepsilon)$, a valuation $v$ such
that $\dom(v)=\dom(\varepsilon)$, and a finite run $\pi$ of a Kripke
structure, let $v+_\varepsilon \pi$ be the valuation describing the
values of variables in $\dom(v)$ \emph{after} following $\pi$ (\ie
once the states of $\pi$ have all been visited and belong to the
past): at each step along $\pi$, the value of every variable $z \in
\dom(v)$ is updated to take into account the truth value of
$\varepsilon(z)$. Formally $v+_\varepsilon \pi$ is defined inductively
as: $v+_\varepsilon \pi = v$ if $|\pi|=0$ (\ie $\pi$ is the empty
sequence), and $(v+_\varepsilon \pi \cdot r)(z) = v'(z)+1$ if $(r,v',
\varepsilon) \sat \varepsilon(z)$ (the satisfaction relation $\sat$ is
defined below) and $(v+_\varepsilon \pi\cdot r)(z) = v'(z)$ otherwise,
where $v'$ is the valuation $v+_\varepsilon \pi$ and $r$ is a state.

\begin{defi}
  The following clauses define the satisfaction of a consistent
  $\CCTLv$ closure $(\vfi,\varepsilon)$ from the state $q$ of some
  Kripke structure $\calS=\tuple{Q,R,\ell}$ under valuation $v$ with
  $\dom(v)=\dom(\varepsilon)$ -- written $(q,v,\varepsilon)
  \sat_{\calS} \vfi$ -- by induction over the structure of $\vfi$ (we
  omit the cases of Boolean modalities):
  \begin{alignat*}{3}
    & (q,v,\varepsilon) \sat_{\calS} z[\psi].\vfi & & \text{\quad iff
      \quad} & & (q,v[z\leftarrow 0],\varepsilon[z\leftarrow \psi])
    \sat_{\calS} \vfi,
    \\
    & (q,v,\varepsilon) \sat_{\calS} \textstyle \sum_{i=1}^{\ell}
    \alpha_i \cdot z_i \sim c & & \text{\quad iff \quad} & &
    \textstyle \sum_{i=1}^{\ell} \alpha_i\cdot v(z_i) \sim c,
    \\
    & (q,v,\varepsilon) \sat_{\calS} \Ex \vfi \U \psi & & \text{\quad
      iff \quad} & & \exists \rho \in \Exec(q)\ \text{s.t.}\
    (\rho,v,\varepsilon) \sat_{\calS} \vfi \U \psi,
    \\
    & (q,v,\varepsilon) \sat_{\calS} \All \vfi \U \psi & & \text{\quad
      iff \quad} & & \forall \rho \in \Exec(q), \text{we have}\
    (\rho,v,\varepsilon) \sat_{\calS} \vfi \U \psi,
    \\
    \intertext{where} & (\rho,v,\varepsilon) \sat_{\calS} \vfi \U \psi
    \quad & & \text{\quad iff \quad} & & \exists i\geq 0\ \text{s.t.}\
    (\rho(i), v +_\varepsilon \rho{|_{i-1}}, \varepsilon) \sat_{\calS}
    \psi
    \\
    & & & \text{\quad and \quad} & & \forall\, 0 \leq j < i,\
    (\rho(j), v +_\varepsilon \rho{|_{j-1}}, \varepsilon) \sat_{\calS}
    \vfi.
  \end{alignat*}  
\end{defi}

When there is no risk of confusion, we may omit subscript $\calS$, and
simply write $(q, v, \varepsilon) \sat \vfi$. For any closed formula
$\Phi$, only the state $q$ is relevant and we will simply write $q
\sat_{\calS} \Phi$, or directly $q \sat \Phi$. Remark that, when
evaluating a closed formula according to the above semantic rules,
only consistent closures are built and considered.

Finally, as a technical tool for the following proofs, we consider the
\emph{set of relevant variables} of a closure, that is the set of
variables whose current value is required to decide whether the
formula holds for a given state. Given a consistent closure
$(\Phi,\varepsilon)$, we define $\RV(\Phi, \varepsilon)$ as follows:
\begin{alignat}{2}
  & \RV(z[\psi].\vfi,\varepsilon) &\ \egdef\ &
  \RV(\vfi,\varepsilon[z\leftarrow \psi]) \backslash \{z\}
  \\
  & \RV(\Ex \vfi_1 \U \vfi_2,\varepsilon) &\ \egdef\ & \RV(\All \vfi_1
  \U \vfi_2,\varepsilon) \egdef\RV(\vfi_1 \ou \vfi_2,\varepsilon)
  \egdef \RV(\vfi_1,\varepsilon) \union \RV(\vfi_2,\varepsilon)
  \\
  & \RV(\non \vfi,\varepsilon) &\ \egdef\ &\RV(\vfi,\varepsilon)
  \\
  & \RV(P,\varepsilon) &\ \egdef\ & \emptyset \\
  & \RV(z_i \sim c,\varepsilon) &\ \egdef\ & \{z_i\} \union
  \RV(\varepsilon(z_i), \varepsilon) \label{eq:2}
\end{alignat}
Note that relevant variables in formula $\psi$ are only added to
$\RV(z[\psi].\vfi,\varepsilon)$ when $z_i$ occurs in formula $\vfi$,
i.e. in case \eqref{eq:2} above. Clearly $\FV(\Psi) \subseteq
\RV(\Psi,\varepsilon) \subseteq \V(\Psi)$. Moreover by
Def. \ref{def:consistent}, for every $z' \in \RV(\varepsilon(z),
\varepsilon)$, $z' \prec z$.

\begin{exa}
  Consider the consistent closure $(\Psi, \varepsilon)$ with 
  \[
  \Psi \egdef z_4[P'].\EF(z_4\geq 2 \et z_2=4) \text{\quad and \quad}
  \varepsilon = \{ z_1 \mapsto P, z_2 \mapsto \EX (z_1>2), z_3 \mapsto
  P'' \},
  \]
  we have $\FV(\Psi)=\{z_2\}$ and $\RV(\Psi,\varepsilon) =
  \{z_2,z_1\}$ because $z_1$ occurs free in $\varepsilon(z_2)$, hence
  $\RV(\varepsilon(z_2), \varepsilon) = \{z_1\}$ by
  Eq. \eqref{eq:2}. Of course $z_3$ belongs to neither set because it
  occurs nowhere, and $z_4$ because it is bound in $\Psi$ and
  $\RV(\varepsilon(z_4), \varepsilon) = \emptyset$.
\end{exa}

Given a closure $(\Psi, \varepsilon)$ and a valuation $v$, we denote
by $v_{\Psi}$ the restriction $v_{|\RV(\Psi,\varepsilon) }$ of $v$ to
the domain $\RV(\Psi, \varepsilon)$ (and $\varepsilon_\Phi$ is the
corresponding restriction of $\varepsilon$). The set
$\RV(\Psi,\varepsilon)$ contains the relevant variables for evaluating
$\Psi$ as stated by the following lemma.

\begin{lem}
  For any consistent closure $(\Psi,\varepsilon)$, the closure
  $(\Psi,\varepsilon_\Psi)$ is consistent. Moreover, let $v$ be a
  valuation over $\dom(\varepsilon)$ and $q$ a state,
  \[
  (q,v,\varepsilon) \sat \Psi \iff (q,v_{\Psi},\varepsilon_\Psi)
  \sat \Psi.
  \]
\end{lem}

The proof of this lemma is straightforward.
In the remainder of this section, we will study the expressiveness of
this logic, as well as the complexity of its model-checking and
satisfiability problems.

\subsubsection{Expressiveness}

Similarly to \CCTL formulas without diagonal constraints, we show in
this section that any closed \CCTLv formula can be translated into an
equivalent $\CTL$ formula. 

\begin{prop}
  For every closed $\CCTLv$ formula $\Phi$, there exists an equivalent
  $\CTL$ formula of dag-size $2^{O(|\Phi|^2)}$.
\end{prop}

Before presenting the actual translation, we show that variable values
may be bounded without changing the satisfaction of a formula. For a
valuation $v$ and an integer $K$, let us denote by $v_K$ the
restriction of $v$ to the domain $\{z \in \dom(v) \mid v(z) \leq K\}$.

\begin{lem}
\label{prop-freeze-fvm}
Let $(\phi,\varepsilon)$ be a consistent \CCTLv closure, and $K$ the
maximal constant occurring in a constraint in $\phi$ or
$\varepsilon$. For all Kripke structure $\calS$, state $q$ of $\calS$
and valuations $v$ and $v'$ over $\dom(\varepsilon)$, we have:
\[
v_K = v'_K \implies \big( (q,v,\varepsilon) \sat_\calS \phi \iff
(q,v',\varepsilon) \sat_\calS \phi \big).
\]
\end{lem}

\begin{proof}
  A reformulation of $v_K = v'_K$ is $v(z) \leq K \impl v(z) =
  v'(z)$. For each free variable $z$ whose value by $v$ is greater
  than $K$, the truth value of any constraint where $z$ occurs will be
  the same for $v(z)$ and any other value greater than $K$, in
  particular $v'(z)$, since the constant in the right-hand side of the
  constraint is at most $K$. This is true in $q$, and remains true
  along any run from $q$.
\end{proof}

For any consistent closure $(\phi,\varepsilon)$ and for some valuation
$v$ with $\dom(v)=\dom(\varepsilon)=\RV(\phi,\varepsilon)$, we define
the $\CTL$ translation $\tr{\phi}^v_\varepsilon$ by induction on the
structure of $\phi$. The case of boolean connectives and atomic
formulas is trivial:
\begin{align}
  \tr{\psi_1 \et \psi_2}^{v}_{\varepsilon} & =
  \tr{\psi_1}^{v_{\psi_1}}_{\varepsilon_{\psi_1}} \et
  \tr{\psi_2}^{v_{\psi_2}}_{\varepsilon_{\psi_2}} &
  \tr{P}^{v}_{\varepsilon} & = P & \tr{\non \phi}^{v}_{\varepsilon} &
  = \non \tr{\phi}^{v}_{\varepsilon} \label{eq:10}
\end{align}
Variable definitions and constraints are also straightforward. It
suffices to update and use the valuation and environment suitably:
\begin{align}
  \tr{z[\phi].\psi}^{v}_{\varepsilon} & = \tr{\psi}^{v[z \leftarrow
    0]}_{\varepsilon[z \leftarrow \phi]} & \tr{\textstyle \sum_i
    \alpha_i\cdot z_i \sim c}^{v}_{\varepsilon} & =
  \begin{cases}
    \top & \text{if } \sum_i \alpha_i\cdot v(z_i) \sim c
    \\
    \bot & \text{otherwise}
  \end{cases} \label{eq:11}
\end{align}
Dealing with temporal modalities is more complex, and justifies the
introduction of auxiliary formulas. Similarly to the translation of
$\CCTL$ to $\CTL$, the idea is to successively evaluate each formula
$\varepsilon(z)$ which is \emph{relevant} to the truth value of the
whole formula, and to update the valuation accordingly.
%
%Note that this is done in increasing variable order (according to the
%order ($\prec$) discussed above). 
However, since variable values strictly larger than $K$ (where $K$ is
the largest constant occurring in the formula or the environment) are
all equivalent according to the previous proposition, it is only
useful to evaluate formulas $\varepsilon(z)$ such that $v(z) \leq K$.
\begin{equation}
  \tr{\Ex \phi \U \psi}^{ v}_{\varepsilon} = \Ex \big(
  \tr{\phi}^{v_{\phi}}_{\varepsilon_{\phi}} \et \Theta^{v}_{\varepsilon} \big) \U
  \Big[ \tr{\psi}^{v_{\psi}}_{\varepsilon_{\psi}} \ou \Big(
  \tr{\phi}^{v_{\phi}}_{\varepsilon_{\phi}} \et \Gamma ^{v}_{\varepsilon} \big( \Ex
  \phi \U \psi, \dom(v_K), v \big) \Big)
  \Big] \label{eq:cctlv-trans-until}
\end{equation}
with $\Theta^{v}_{\varepsilon} = \ET_{z \in \dom(v_K)} \big( \non
\tr{\varepsilon (z)}^{v_{\varepsilon
    (z)}}_{\varepsilon_{\varepsilon(z)}} \big)$ and, for $Z \neq
\emptyset$, $z \in Z$ and $v'$ a valuation:
\begin{multline}
  \Gamma^{ v}_{\varepsilon}(\Ex \phi \U \psi, Z,v') = \big( \non
  \tr{\varepsilon(z)}^{
    v_{\varepsilon(z)}}_{\varepsilon_{\varepsilon(z)}} \et \Gamma^{
    v}_{\varepsilon}(\Ex \phi \U \psi, Z \setminus \{z\},v') \big)
  \\
  \ou \big(
  \tr{\varepsilon(z)}^{v_{\varepsilon(z)}}_{\varepsilon_{\varepsilon(z)}}
  \et \Gamma^{ v}_{\varepsilon}(\Ex \phi \U \psi, Z \setminus
  \{z\},v'[z \leftarrow v(z)+1]) \big) \label{eq:13}
\end{multline}
and finally:
\begin{equation}
  \Gamma^{v}_{\varepsilon}(\Ex \phi \U \psi, \emptyset,v') = \left\{
  \begin{aligned}
    & \bot & & \text{ if } v = v',
    \\
    & \EX \tr{\Ex \phi \U \psi}^{v'}_{\varepsilon} & & \text{ otherwise.}
  \end{aligned}\right.
\label{eq:9}
\end{equation}
Finally, given a closed \CCTLv formula $\Phi$, we define its \CTL
translation $\tr{\Phi}$ as $\tr{\Phi}^\emptyset_\emptyset$.

Intuitively, the above translation of \emph{until} modalities with
valuation $v$ and environment $\varepsilon$ works by distinguishing
\emph{interesting} states, in which the value of at least one variable
in $\dom(v_K)$ changes, from uninteresting ones. The \CCTLv formula
$\Ex \phi \U \psi$ then holds if, and only if, after a finite sequence
of uninteresting states satisfying $\vfi$, either $\psi$ holds or the
run has reached an interesting state satisfying $\vfi$, after which
$\Ex \phi \U \psi$ holds with a suitably updated valuation.

Formula $\Theta^{v}_{\varepsilon}$ in Eq. \eqref{eq:cctlv-trans-until}
expresses the fact that the current state is uninteresting, and
$\Gamma^{ v}_{\varepsilon}(\Ex \phi \U \psi, \dom(v_K), v)$ that the
current state is interesting, in other words satisfies at least one of
the formulas $\varepsilon(z)$ for $z$ a variable with value at most
$K$ in $v$, and satisfies $\EX \tr{\Ex \phi \U
  \psi}^{v}_{\varepsilon}$. For such a state it is necessary to know
exactly which formulas $\varepsilon(z)$ are satisfied and this is done
by scanning the set $\dom(v_K)$, updating the valuation $v'$ for each
$z$ in turn whenever $\varepsilon(z)$ is attested to hold
(Eq. \eqref{eq:13}). If no $\varepsilon(z)$ holds in the current
state, which is witnessed by the fact that $v = v'$, the state is in
fact uninteresting and the whole scanning fails, otherwise the
unfolding process continues (Eq. \eqref{eq:9}).

Note how $v$ and $\varepsilon$ are restricted to relevant variables at
every recursive call to the above translation procedure (for instance
in
$\tr{\varepsilon(z)}^{v_{\varepsilon(z)}}_{\varepsilon_{\varepsilon(z)}}$). This
precaution is used to avoid cycles in the update of variables. It is
necessary, since simply translating $\varepsilon(z)$ with an
environment and valuation containing $z$ itself may generate an
infinite formula. It is also sufficient, since by definition of
consistent closures, $z \not\in \RV(\varepsilon(z),\varepsilon)$.

Formulas $\tr{\All \phi \U \psi}^{ v}_{\varepsilon}$ and $\Gamma^{
  v}_{\varepsilon}(\All \phi \U \psi, Z,v')$ are defined similarly by
replacing each occurrence of $\Ex$ with $\All$ in the above formulas.

\begin{lem}
  The above inductive definition for $\tr{\Phi}$ is well-founded, in
  other words $\tr{\Phi}$ is a finite \CTL formula. The DAG-size of
  $\tr{\Phi}$ is in $2^{O(|\Phi|^2)}$. 
\end{lem}

\begin{proof}
  In Equations (\ref{eq:10}) and (\ref{eq:11}), all inductive uses of
  the translation function are performed over strictly shorter
  formulas. Even though this is not the case in Eq. (\ref{eq:9}), no
  recursive call is made unless the valuation $v'$ used in
  Eq. (\ref{eq:9}) is different from (hence necessarily strictly
  greater than) $v$. Since variables assigned a value greater than $K$
  do not belong to $\dom(v_K)$, this set will eventually become empty,
  meaning that no state is considered interesting after some
  point. Hence no infinite inductive ``call'' to $\tr{\Ex \phi \U
    \psi}^{v}_{\varepsilon}$ is possible. Finally, the definition of
  $\Gamma^{ v}_{\varepsilon}(\phi, Z, v')$ only refers to formulas
  $\Gamma^{ v}_{\varepsilon}(\phi, Z',v')$ with $Z'$ strictly included
  in $Z$.

  The maximal number of distinct valuations $v$ we need to consider is
  bounded by $(K+3)^{n}$ (since each of the $n$ variables can assume a
  value between $0$ and $K+1$ or be undefined). Since each $\Gamma^{
    v}_{\varepsilon}(\phi, Z, v')$ is indexed by two valuations $v$
  and $v'$, one sub-formula $\vfi$ (of which there are at most
  $|\Phi|$) and a set of variables $Z$ (at most $2^n$ possibilities),
  the total number of distinct such formulas to consider is less than
  $((K+3)^{n})^2 \cdot |\Phi| \cdot 2^n$. Overall, since $K \in
  O(2^{|\Phi|})$ due to the binary encoding and $n \in O(\Phi)$, this
  yields a worst-case DAG-size for $\tr{\Phi}$ in $O(|\Phi| \cdot
  (2^{|\Phi|}+3)^{2|\Phi|} \cdot 2^{|\Phi|}) \subseteq
  2^{O(|\Phi|^2)}$.
\end{proof}
 
We have the following correctness lemma:

\begin{lem}
  Let $(\Phi,\varepsilon)$ be a consistent \CCTLv closure, $K$ the
  maximal constant in $\Phi$ and $\varepsilon$. For every Kripke
  structure $\calS$, state $q$ of $\calS$ and $(K+1)$-bounded
  valuation $v$ we have:
  \[
  (q, v, \varepsilon) \sat_{\calS} \Phi \iff q \sat_{\calS}
  \tr{\Phi}^{v_{\Phi}}_{\varepsilon_{\Phi}}.
  \]
\end{lem}

\begin{proof}
  The proof of the direct implication is done by structural induction
  over $\Phi$. We only detail the cases of variable definition and
  temporal modalities.
  \begin{iteMize}{$\bullet$}
  \item $\Phi \egdef z[\phi] . \psi$: Assume $(q,v,\varepsilon) \sat
    z[\phi] . \psi$. This is semantically equivalent to $(q,v[z
    \leftarrow 0],\varepsilon[z \leftarrow \phi]) \sat \psi$. By
    induction hypothesis $q \sat \tr{\psi}^{v[z \leftarrow
      0]_\psi}_{\varepsilon[z \leftarrow \phi]_\psi}$, hence $q \sat
    \tr{\Phi}^{v_{\Phi}}_{\varepsilon_{\Phi}}$.
  \item $\Phi \egdef \Ex \phi \U \psi$: Assume $(q,v,\varepsilon) \sat
    \Ex \phi \U \psi$. There exists a run $\rho = q_0 q_1 q_2 \ldots$
    with $q_0 = q$ and an index $i \geq 0$ such that $(q_i, v
    +_\varepsilon \rho{|_{i-1}}, \varepsilon) \sat \psi$ and for all
    $0 \leq j < i$, $(q_j, v+_\varepsilon \rho{|_{j-1}}, \varepsilon)
    \sat \phi$.
    For every $0 \leq j < i$, let $v_j$ be the valuation
    $v_\Phi+_{\varepsilon_\Phi} \rho{|_{j-1}}$, and $Z_j$ be the set
    of variables $z$ such that $v_j(z)\leq K$ and $v_{j+1}(z) = v_j(z)
    + 1$, i.e. the set of relevant variables whose value is
    incremented in state $q_j$.
    \\
    Let $j_1$,\ldots, $j_\ell$ be the positions in $\{0,\ldots,i-1\}$
    along $\rho$ where $Z_{j_h}$ is non-empty. We reason by induction
    over $\ell$. If $\ell=0$, then clearly $q \sat \Ex (
    \tr{\phi}^{v_\phi}_{\varepsilon_\phi} \et
    \Theta^{v_\Phi}_{\varepsilon_\Phi} ) \U
    \tr{\psi}^{v_\psi}_{\varepsilon_\psi}$, and thus $q \sat
    \tr{\Phi}^{v_\Phi}_{\varepsilon_\Phi}$. Now assume $\ell>0$, we
    have: $Z_{j} = \emptyset$ for $0\leq j < j_1$, $Z_{j_1} \not=
    \emptyset$, and:
    \[
    q_{j_1} \sat \underbrace{\ET_{z \in Z_{j_1}}\varepsilon_\Phi(z)}_{
      \Phi_1} \:\et\: \underbrace{\ET_{z \in
        \dom(\varepsilon_\Phi)\backslash Z_{j_1}} \non
      \varepsilon_\Phi(z)}_{ \Phi_2}
    \]
    Moreover we have $(q_{j_1+1},v_{j_1+1},\varepsilon) \sat \Ex \phi
    \U \psi$.  By induction hypothesis over $\ell$ we have $q_{j_1+1}
    \sat \tr{\Ex \vfi \U \psi}^{v_{j_1+1}}_{\varepsilon_\Phi}$ and
    thus: $q \sat \Ex ( \tr{\phi}^{v_\phi}_{\varepsilon_\phi} \et
    \Theta^{v_\Phi}_{\varepsilon_\Phi} ) \U (\Phi_1 \et \Phi_2 \et \EX
    \tr{\Ex \phi \U
      \psi}^{(v_{j_1+1})}_{\varepsilon_\Phi})$. Therefore we have $q
    \sat \tr{\Phi}^{v_\Phi}_{\varepsilon_\Phi}$.
  \item $\Phi \egdef \All \phi \U \psi$: in this case, every run from
    $q$ has to verify $\phi\U\psi$. We can reuse the same approach as
    before. In the general case, every run starts with a prefix along
    which every state $q_j$ is such that $Z_j$ is empty, followed by
    some state $q_{j_1}$ where $Z_{j_1} \neq \emptyset$, which
    satisfies $\AX\tr{\All \phi \U
      \psi}^{(v_{j_1+1})_\Phi}_{\varepsilon}$.
  \end{iteMize}

  \noindent
  The converse is also done by structural induction on $\Phi$. The
  case where $\Phi \egdef z[\phi] . \psi$ follows the same reasoning
  as above, only backwards. 
  When $\Phi \egdef \Ex \phi \U \psi$, we reason by induction on the
  following (well-founded) ordering of valuations. We write $v' \unlhd
  v$ whenever $\dom(v'_K) \subseteq \dom(v_K)$ and $\forall x \in
  \dom(v'_K), v'(x) \geq v(x)$, meaning that $v'$ assigns greater
  values than $v$ to all variables to which $v'$ assigns a value less
  than or equal to $K$, and $v' \lhd v$ if additionally $v' \neq
  v$. Assume $q \sat \tr{\Phi}^{v}_{\varepsilon}$, and consider the
  iterative unfolding of the definitions of subformula $\Gamma$ in
  $\tr{\Phi}^{v}_{\varepsilon}$. For this formula to hold, there must
  exist a satisfied formula $\Psi$, obtained by replacing each
  disjunction by one of its operands, resulting in a ``witness'' for
  the satisfaction of $\tr{\Phi}^v_\varepsilon$. $\Psi$ is of one of
  the forms:
 \begin{equation}
   \label{eq:7}
   \Psi = \Ex \big( \tr{\phi}^{v_{\phi}}_{\varepsilon_{\phi}} \et
   \Theta^{v}_{\varepsilon} \big) \U \big(
   \tr{\phi}^{v_{\phi}}_{\varepsilon_{\phi}} \et \ET_{z \in Z}
   \tr{\varepsilon(z)}^{v_{\varepsilon(z)}}_{\varepsilon_{\varepsilon(z)}}
   \et \ET_{z \in \dom(v_K) \setminus Z}
   \non\tr{\varepsilon(z)}^{v_{\varepsilon(z)}}_{\varepsilon_{\varepsilon(z)}}
   \et \EX \tr{\Phi}^{v'}_{\varepsilon} \big)
  \end{equation}
   for some non-empty $Z \subseteq \dom(v_K)$, and with $v'(z) = v(z) +
  1$ if $z \in Z$ and $v(z) \leq K$ and $v'(z) = v(z)$ otherwise, or
  \begin{equation}\label{eq:12}
    \Psi = \Ex \big( \tr{\phi}^{v_{\phi}}_{\varepsilon_{\phi}} \et
    \Theta^{v}_{\varepsilon} \big) \U
    \tr{\psi}^{v_{\psi}}_{\varepsilon_{\psi}}.
  \end{equation}
  In the former case (Eqn. (\ref{eq:7})), there must exist a run $\rho
  = q_0 q_1 \ldots$ and some $k \geq 0$ such that $q_i
  \not\sat_{\calS}
  \tr{\varepsilon(z)}^{v_{\varepsilon(z)}}_{\varepsilon_{\varepsilon(z)}}$
  for all $i < k$ and $z \in \dom(v_K)$, $q_i \sat_{\calS}
  \tr{\phi}^{v_{\phi}}_{\varepsilon_{\phi}}$ for all $i \leq k$, $q_k
  \sat_{\calS}
  \tr{\varepsilon(z)}^{v_{\varepsilon(z)}}_{\varepsilon_{\varepsilon(z)}}$
  for all $z \in Z$, $q_k \not\sat_{\calS}
  \tr{\varepsilon(z)}^{v_{\varepsilon(z)}}_{\varepsilon_{\varepsilon(z)}}$
  for all $z \in \dom(v_K) \setminus Z$, and $q_{k+1} \sat_{\calS}
  \tr{\Phi}^{v'}_{\varepsilon}$.

  Since $Z \neq \emptyset$, we have $v' \lhd v$, hence by our
  induction hypotheses over the structure of $\Phi$ and the ordering
  of valuations, we obtain that $(q_i, v, \varepsilon)
  \not\sat_{\calS} \varepsilon(z)$ for all $i < k$ and $z \in
  \dom(v_K)$, $(q_i, v, \varepsilon) \sat_{\calS} \phi$ for all $i
  \leq k$, $(q_k, v, \varepsilon) \sat_{\calS} \varepsilon(z)$ for all
  $z \in Z$, $(q_k, v, \varepsilon) \not\sat_{\calS} \varepsilon(z)$
  for all $z \in \dom(v_K) \setminus Z$, and $(q_{k+1}, v',
  \varepsilon) \sat_{\calS} \Phi$.

  Since the truth value of any subformula is independent of the
  variables which are irrelevant for that subformula or whose value is
  already greater than $K$ at the beginning of the run, and given the
  truth values of formulas $\varepsilon(z)$ along
  $\rho$, this implies that $(q_{k+1}, v +_\varepsilon \rho|_{k},
  \varepsilon) \sat_{\calS} \Ex \phi \U \psi$ and $\forall i
  \leq k, (q_{i}, v +_\varepsilon \rho|_{i-1}, \varepsilon)
  \sat_{\calS} \phi$, hence $(q_0, v, \varepsilon) \sat_\calS \Phi$,
  and this remains true for any valuation $v''$ and environment
  $\varepsilon''$ such that $v''_\Phi = v$ and $\varepsilon''_\Phi =
  \varepsilon$.

  The latter case (Eqn. (\ref{eq:12})) is easier and is solved
  similarly.  As previously, the $\All$ quantifier is also treated
  in the same fashion.
\end{proof}

\begin{exa}
  For the \CCTLv formula $\Phi = z[P] . z'[z>0] . \EF(z'>0 \et
  P')$, we obtain (after simplification) the following translation:
  \[
  \tr{\Phi} \;\eqdef\; \Ex \: (\non P) \: \U \Big(P \et \EX \big( \EX
  \: \EF P' \big) \Big)
  \]
  The two nested $\EX$ modalities are necessary because one must
  distinguish the first state $r$ where $P$ holds true from its
  successor $r'$, which is the first to satisfy $z>0$, and from the
  successor $r''$ of $r'$ which is the first state satisfying $z'>0$.
\end{exa}

\subsubsection{Model checking}

\begin{thm}
\label{theo-freeze-hard}
Model checking closed $\CCTLv$ formulas is \PSPACE-complete.
\end{thm}

\proof\hfill
\begin{iteMize}{$\bullet$} 
\item
\PSPACE-hardness can be proved by a reduction from the
quantified Boolean formula problem (QBF)\footnote{This is a
  simplification of the reduction used for $\TCTL_c$ over KS
  \cite{lst-TCS2001}.}.  Consider a QBF instance $ \calI \egdef
\exists x_1 \forall x_2 \ldots \exists x_{2p-1} \forall x_{2p} \cdot
\Phi$ where $\Phi$ is a propositional formula in 3-conjunctive normal
form (3-CNF) $\ET_{j=1\ldots m} (\ell^j_1 \ou \ell^j_2 \ou \ell^j_3)$
over $\{x_1,\ldots,x_{2p}\}$. Now consider the KS $\calS_\calI =
\tuple{Q,R,\ell}$ in Figure ~\ref{fig-freeze-qbf}.  We assume that
every state $q_i$ is labeled with its name, and every state $x_i$
(resp.\ $\bar{x_i}$) is labeled by the atomic proposition $C_j$ iff
$x_i$ (resp.\ $\non x_i$) is one of the literals in
$\{\ell^j_1,\ell^j_2,\ell^j_3\}$. Then $\calI$ is positive iff $q_1$
satisfies the following formula:
\[
 z_1[C_1]\cdots z_{m}[C_m]\cdot \EF
\bigg(q_2 \et \AF \Big(q_3 \et \EF \ldots \big(q_{2p} \et \AF (q_{2p+1} \et
\ET_{i=1\ldots m} (z_i \geq 1) ) \big) \Big) \bigg)
\]
% anciene version
% $\bullet$ \PSPACE-hardness can be proved by a reduction from
% QBF\footnote{This is a simplification of the reduction used for
%   $\TCTL_c$ over KS \cite{lst-TCS2001}.}.  Consider a QBF instance $
% \calI \egdef \exists x_1 \forall x_2 \ldots \exists x_{2p-1} \forall
% x_{2p} \cdot \Phi$ where $\Phi$ is a propositional formula over
% $\{x_1,\ldots,x_{2p}\}$. Now consider the KS $\calS_\calI =
% \tuple{Q,R,\ell}$ in Figure ~\ref{fig-freeze-qbf}.  We assume that
% every state $x_i$ or $q_i$ is labeled with its name and there is no
% labeling for states $\bar{x_i}$. Then $\calI$ is positive iff $q_1$
% satisfies the property $z_1[x_1]\cdots z_{2p}[x_{2p}]\cdot \EF (q_2
% \et \AF (q_3 \et \EF \ldots \AF (q_{2p+1} \et \widetilde{\Phi})\ldots)))$,
% where $\widetilde{\Phi}$ is $\Phi[x_i \leftarrow "z_i=1";\bar{x}_i
% \leftarrow "z_i=0"]$.

\begin{figure}[t]
  \centering
  \begin{tikzpicture}[every join/.style={->}]
    \scriptsize

     \matrix[row sep=2.5mm, column sep= 1mm] {
      & \node (x1) {$x_1$}; & & \node (x2) {\makebox[1em][c]{$x_2$}}; & & &
      & \node (x2p-1) {$x_{2p-1}$}; & & \node (x2p) {$x_{2p}$}; \\

      \node[initial,initial text={}] (q1) {$q_{1}$}; & & \node (q2) {\makebox[1em][c]{$q_{2}$}}; 
      & & \node (q3) {\makebox[1em][c]{$q_{3}$}}; 
      & \node {$\cdots$}; &
      \node (q2p-1) {$q_{2p-1}$};  & & \node (q2p) {$q_{2p}$};  
      & & \node (q2p+1) {$q_{2p+1}$}; \\

      & \node (nx1) {$\bar{x}_1$}; & & \node (nx2) {\makebox[1em][c]{$\bar{x}_{2}$}}; & & &
      & \node (nx2p-1) {$\bar{x}_{2p-1}$}; & & \node (nx2p) {$\bar{x}_{2p}$};  \\
     };

     { [start chain] 
       \chainin (q1); \chainin (x1) [join]; \chainin (q2) [join];
       \chainin (x2) [join]; \chainin (q3) [join]; \chainin (q2p-1);
       \chainin (x2p-1) [join]; \chainin (q2p) [join]; \chainin (x2p)
       [join]; \chainin (q2p+1) [join];
     }

     { [start chain] 
       \chainin (q1); \chainin (nx1) [join]; \chainin (q2) [join];
       \chainin (nx2) [join]; \chainin (q3) [join]; \chainin (q2p-1);
       \chainin (nx2p-1) [join]; \chainin (q2p) [join]; \chainin (nx2p)
       [join]; \chainin (q2p+1)  [join]; 
     }

     \path (q2p+1) edge [->,in=-30,out=30,loop] (q2p+1);

  \end{tikzpicture}
  \caption{Kripke structure associated to a QBF instance over
    $\{x_1,\ldots,x_{2p}\}$.}
  \label{fig-freeze-qbf}
\end{figure}
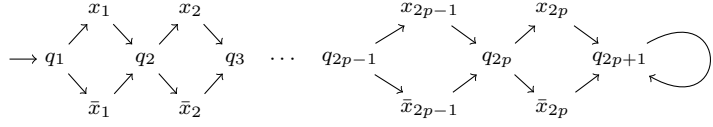

\item \PSPACE-membership is obtained by considering a
non-deterministic algorithm working in polynomial space to decide
whether a closed $\CCTLv$ formula $\Phi$ holds for a state $q$ within
a KS $\calS$. This provides an \NSPACE\ procedure which, by Savitch's
theorem, implies the existence of a \PSPACE\ algorithm.  

We assume that $\Phi$ contains $n$ variables $z_1, z_2, \ldots
z_n$. Let us call \emph{configuration} any triple $(q, v,
\varepsilon)$ where $q$ is a state, $v$ a valuation and $\varepsilon$
an environment. First note that valuations can be encoded in space
polynomial in $|\Phi|$ since it is sufficient to store the value for
each variable $z$ as a $K+1$-bounded counter, where $K$ is the maximal
constant occurring in $\Phi$, which requires at most $|\Phi|$ bits per
variable. Hence configurations can be encoded in space polynomial in
$|\Phi|$ and linear in $|\calS|$.

For any consistent closure $(\Psi, \varepsilon)$ with $\Psi \in
\subf(\Phi)$, we define an \NSPACE\ procedure $\Check(q, v,
\varepsilon, \Psi)$ to decide whether $\Psi$ holds over $(q, v,
\varepsilon)$.
% Let us define runs over configurations as infinite sequences of
% configurations whose underlying sequence of states forms a run in
% $\calS$. Such a run is called a witness for (the satisfaction of)
% $(\Psi, \varepsilon)$ if its sequence of valuations is consistent
% with $\epsilon$, meaning that each variable whose associated formula
% is satisfied in a given state $q$ is incremented after each
% configuration $(q, v, \epsilon)$.
We consider several cases according to the structure of $\Psi$, of
which we omit the simplest.

\begin{iteMize}{$-$}
\item $\Psi \egdef z_i[\psi_i].\vfi_i$ : the returned value is
  $\Check(q, v[z_i\leftarrow 0], \varepsilon[z_i \leftarrow \psi_i],
  \vfi_i)$.
\item $\Psi \egdef \sum_{i=1}^\ell \alpha_i \cdot z_i \sim c$ : the
  returned value is the boolean evaluation of the constraint
  $\sum_{i=1}^\ell \alpha_i \cdot v(z_i) \sim c$.
\item $\Psi \egdef \Ex \vfi_1 \U \vfi_2$: if $\Check(q, v,
  \varepsilon, \vfi_2)$ is evaluated to $\top$, then the returned
  value is $\top$. Else if $\Check(q, v, \varepsilon, \vfi_1)$ is
  $\bot$, then the result is $\bot$. Otherwise we proceed as follows:
  \begin{enumerate}[(1)]
  \item for every $z_i \in \RV(\Psi, \varepsilon)$, call $\Check(q,
    v_{\varepsilon(z_i)}, \varepsilon_{\varepsilon(z_i)},
    \varepsilon(z_i))$ and assign $1$ to an integer variable
    $\delta_i$ if the result is $\top$, and $0$ otherwise;
  \item guess a transition $q \rightarrow q'$ in $\calS$;
  \item replace the current configuration $(q, v, \varepsilon)$ by
    $(q', v', \varepsilon)$ by $v'(z_i) = \min(K+1, v(z_i) +
    \delta_i)$ for all $z_i \in \RV(\Psi, \varepsilon)$, and check
    whether $\vfi_2$ holds for it, and so on.
  \end{enumerate}
\item $\Psi \egdef \EG \vfi$: since there are finitely many
  $K+1$-bounded configurations, if there is a run for $\Psi$ starting
  in $q$ then there must also exist one whose corresponding sequence
  of bounded configurations is ultimately periodic, i.e. consists of a
  finite sequence of configurations followed by an infinite repetition
  of a finite configuration cycle (up to valuation equivalence). The
  procedure $\Check(q, v, \varepsilon, \Psi)$ can thus consist of the
  following steps:
  \begin{enumerate}[(1)]
  \item start guessing a sequence of transitions as in the previous
    case, updating the current state and valuation accordingly;
  \item in each new configuration $(q, v, \varepsilon)$, verify that 
   $\Check(q, v, \varepsilon, \vfi)$ is $top$;
%     then return $\bot$, otherwise proceed;
  \item at some point, non-deterministically assume the current
    (bounded) configuration to occur infinitely often in some
    ultimately periodic run satisfying $\Psi$, and store the
    corresponding state $q_r$ and valuation $v_r$;
  \item resume guessing transitions, checking at each step that
    $\Check(q, v, \varepsilon, \vfi)$ is $\top$;
  \item return $\top$ if the previously stored recurring configuration
    is ever encountered again.
  \end{enumerate}
\end{iteMize}
Deciding $q \sat \Phi$ is then achieved by calling $\Check(q, v_0,
\varepsilon, \Phi)$.

The space used by $\Check(q, v_0, \varepsilon, \Phi)$ is evaluated as
follows: for $\Ex \vfi_1 \U \vfi_2$ or $\EG \vfi_1$, we need to store
at most three configurations $(q, v, \varepsilon)$ and $k$ boolean
values. We also need space for the recursive calls over
subformulas. The maximal number of such nested calls is
bounded\footnote{where $\TH(\vfi)$ is the temporal height of $\vfi$
  defined as usual except for the reset operator for which we have:
  $\TH(z[\psi].\vfi)=\TH(\vfi)$.}  by $\TH(\Phi) + \sum_{i=1}^{n}
\TH(\varepsilon(z_i))$: indeed the first term comes from the recursive
calls for $\Check(q,v,\vfi_i)$ and the second from the calls
$\Check(q,v,\varepsilon(z_i))$. Thus the maximal number of nested
calls is bounded by $|\Phi|$. \qed
\end{iteMize}

\begin{rem}
  As soon as subtractions are allowed in $\CCTLv$, model checking
  becomes undecidable as a simple consequence of
  Thm. \ref{thm:mc-BCS2} and Rem. \ref{rem:cctl-to-cctlv}.
\end{rem}

\subsubsection{Satisfiability} 

As in the case of $\CCTL$, the translation of $\CCTLv$ formulas into
$\CTL$ provides an optimal decision procedure for satisfiability:

\begin{thm}
  The satisfiability problem for $\CCTLv$ is 2-EXPTIME-complete.
  \label{satisf-cctlv-2exp}
\end{thm}

\proof A closed $\CCTLv$ formula $\Phi$ is satisfiable (\ie\ it holds
for a state $q$ in a finite KS $\calS$) iff the $\CTL$ formula
$\tr{\Phi}$ is satisfiable. The (DAG) size of $\tr{\Phi}$ is in
$2^{O(|\Phi|^2)}$, which yields a 2EXPTIME procedure to decide
satisfiability of $\Phi$.  Hardness is a consequence of
Thm. \ref{satisf-2exp} and Rem. \ref{rem:cctl-to-cctlv}.  \qed

\subsection{Cumulative semantics for  \texorpdfstring{\CCTL}{CCTL}}

We now define a variant of $\CCTL$ based on an alternative semantics
for $\Ex\_\U\_$ and $\All\_\U\_$ modalities. In this semantics,
nesting two temporal modalities no longer resets the counting process
for the evaluation of the innermost modality: its constraints are then
interpreted over the whole run. In order to relax this semantics, we
add the modality $\Now$ (for 'now', or rather 'from now on') which
specifies that the counters have to be reset in the current state and
start counting again from the current position. Let us fix the syntax
of $\CCTLc$:

\begin{defi}
  \label{sec:def-cctlc}
  Given a set of atomic propositions $\AP$,
  % and a set of constraints $\calC \in \Lcons,$
  we define:
  \[
  \CCTLc \ni \vfi,\psi \grameg P \gramou \vfi \et \psi \gramou \non
  \vfi \gramou \Now \vfi \gramou \Ex \vfi \Usc{C} \psi \gramou \All
  \vfi \Usc{C} \psi
  \]
  with $P\in \AP$. % and $C\in \calC$.
  As in the case of \CCTL, we use shorthands $\Fsc{C} \vfi$ and
  $\Gsc{C} \vfi$ to denote $\top \Usc{C} \vfi$ and $\non \Fsc{C} \non \vfi$
  respectively.
\end{defi}

$\CCTLc$ formulas are interpreted over pairs $(\pi,q)$ where $q$ is a
state of some Kripke structure $\calS$ and $\pi$ is a history
(\emph{i.e.}\ a finite prefix) such that $\pi \cdot q \in
\Execf(\calS)$. The following clauses\footnote{As previously, we only
  give the formal semantics of the main modalities. Boolean
  connectives are interpreted in a natural way.}  define when a
$\CCTLc$ formula $\Phi$ holds for $(\pi,q)$:
  \[
  \begin{array}{lcl}
    (\pi,q) \sat_\calS \Ex \vfi \Usc{C} \psi & \ \mbox{iff} \ \exists \rho \in
    \Exec(q), \: \exists i\geq 0, & \ (\pi\cdot\rho_{|i-1},\rho(i)) \sat_\calS \psi, \;
    \pi\cdot \rho_{|i-1} \sat_\calS C, 
    \\
    & & \mbox{ and } \forall 0 \leq j < i, \: (\pi\cdot\rho_{|j-1},\rho(j)) \sat_\calS \vfi 
    \\
    (\pi,q) \sat_\calS \All \vfi \Usc{C} \psi & \ \mbox{iff} \ \forall \rho \in
    \Exec(q), \: \exists i\geq 0, & \ (\pi\cdot\rho_{|i-1},\rho(i)) \sat_\calS \psi, \; 
    \pi\cdot\rho_{|i-1} \sat_\calS C, 
    \\
    & & \mbox{ and } \forall 0 \leq j < i, \:
    (\pi\cdot\rho_{|j-1},\rho(j)) \sat_\calS \vfi 
    \\
    (\pi,q) \sat_\calS \Now \vfi & \multicolumn{2}{l}{\ \mbox{iff} \
      (\epsilon,q) \sat_\calS \vfi} 
    \\
  \end{array}
\]

The addition of the $\Now$ modality allows us to easily express \CCTL
properties. Indeed each $\CCTL$ formula $\Phi$ can be easily
translated into a $\CCTLc$ formula $\Psi$ by guarding each of its
temporal modalities with $\Now$. Both formulas are equivalent, in the
sense that for any state $q$ and history $\pi$, we have $q \models
\Phi \iff (\pi, q) \models \Psi$. We also have the following useful
property:
\begin{equation}
  \label{equiv-constraint-cumul}
  (\pi,q) \sat \Ex \bot \Usc{C}
  \top  \iff \pi \sat C 
\end{equation}
For simplicity, in the following we will thus allow ourselves to
directly write constraints in the formula and not only as subscripts
of temporal modalities.

\medskip

\begin{exa}
  The $\CCTLc$ formula $\Ex \Fsc{\nb \top \leq k_1} (P_1 \et \Ex
  \Fsc{\nb\top \leq k_2} P_2)$ with $k_1\leq k_2$ holds for a state
  $q$ if and only if there exists a run with less than $k_2$
  transitions leading to some state satisfying $P_2$ and along this
  run there is a state satisfying $P_1$ located at less than $k_1$
  transitions from $q$.
\end{exa}
  
\begin{exa}
  The $\CCTLc$ formula $\Ex \Fsc{\nb\vfi \geq k_1} \Ex
  \Fsc{\nb\vfi\leq k_2} \psi$ is semantically equivalent to the
  $\CCTLba$ formula $\EFs{k_1 \leq \nb \vfi \leq k_2} \psi$.
\end{exa}

\begin{prop}
  Model checking $\CCTLc$ is \PSPACE-hard.
\end{prop}

\proof We reduce the QBF problem to a model-checking problem for
$\CCTLc$ by using exactly the same reduction as for $\CCTLv$
(Thm. \ref{theo-freeze-hard}): given an instance $\calI$ of QBF, we
consider the same KS $\calS_\calI$ and the following formula:
\[
\EF \bigg(q_2 \et \AF \Big(q_3 \et \EF \ldots \AF \big(q_{2p+1} \et \ET_{i=1\ldots
  m} (\nb C_i \geq 1) \big)\Big)\bigg)
\]
Recall that we can use constraints directly inside the formula due
to the equivalence \eqref{equiv-constraint-cumul} above. \qed

Note that we do not use $\Now$ to prove \PSPACE-hardness. To prove
membership in \PSPACE, we show that one can translate any $\CCTLc$
formula $\phi$ into an equivalent (and succinct) $\CCTLv$ formula
$\overline{\phi}$.  First given $\phi \in \CCTLc$, we use
$\SFnb{\phi}$ to denote the set of subformulas $\psi$ of $\phi$ such
that $\nb\psi$ occurs in a counting constraint inside $\phi$. We now
define $\overline{\Phi}$ as follows:
\[
\begin{array}{l}
  \overline{P}  \; \egdef \; P  
  \hspace{3cm}
  \overline{\vfi \et
    \psi}   \;\egdef\;  \overline{\vfi} \et \overline{\psi}
  \hspace{3cm}
  \overline{\non \vfi} \; \egdef \; \non \overline{\vfi}   
  \\[1.1ex]
  \overline{\Now \vfi} \; \egdef \;
  z_{\psi_1}[\overline{\psi_1}]. \:\ldots\:
  z_{\psi_k}[\overline{\psi_k}].\overline{\vfi} \;\;\; \mbox{with}\;
  \SFnb{\vfi} = \{\psi_1,\ldots,\psi_k\} 
  \\[1.1ex] 
  \overline{\Ex \vfi \Usc{C} \psi} \; \egdef \;  \Ex \overline{\vfi}
  \U (\overline{C} \et \overline{\psi})  
  \hspace{3cm}  
  \overline{\All \vfi \Usc{C} \psi}   \; \egdef \; \All
  \overline{\vfi} \U (\overline{C} \et \overline{\psi})   
  \\[1.2ex]
  {\overline{\sum_i \alpha_i \cdot \nb\vfi_i \:\sim c}  \; \egdef \;  
    \sum_i \alpha_i \cdot z_{\vfi_i}   \:\sim\: c} 
  \\ 
\end{array}
\]
Given a set of formulas $S$, a prefix $\pi$, a valuation $v$ for a
set of variables $V$ and an environment $\varepsilon$, we say that
$(v,\varepsilon)$ is \emph{compatible} with $(S,\pi)$ (written
$(v,\varepsilon) \unrhd (S,\pi)$ ) if and only if for any $\psi\in
S$, there is some $z_\psi \in \dom(v)$ such that $\varepsilon(z_\psi)
= \psi$ and $v(z_\psi) = |\pi|_{\psi}$.

We have the following property:

\begin{lem}
  Let $\Phi$ be a $\CCTLc$ formula, $q$ a state in some KS $\calS$,
  and $\pi \in \Execf(\calS)$ be a finite run such that $\pi\cdot q
  \in \Execf(\calS)$. Let $v: V \fleche \Nat \cup \{\bot\}$ be a
  valuation for a set of variables $V$ and let $\varepsilon$ be an
  environment such that $(v,\varepsilon)$ is compatible with
  $(\SFnb{\Phi},\pi)$.
  % : for any $\psi\in \SFnb{\Phi}$ there is some $z_\psi\in V$ such
  % that $\varepsilon(z_\psi) = \psi$ and $v(z_\psi)= |\pi|_{\psi}$.
  Then:
  \[
  (\pi,q) \sat_{\calS} \Phi \iff (q,v,\varepsilon) \sat_{\calS}
  \overline{\Phi}
  \]
\end{lem}

\begin{proof}
  The proof is done by structural induction over $\Phi$. The result is
  direct for atomic propositions and boolean connectives.

  Let $\Phi \egdef \Ex \vfi \Usc{C} \psi$, and assume $(\pi,q) \sat
  \Phi$. Then there exist $\rho \in \Exec(q)$ and $i\geq 0$ such that
  (a) $(\pi\cdot \rho_{|i-1}, \rho(i)) \sat \psi$, (b)
  $\pi\cdot\rho_{|i-1} \sat C$ and (c) for all $0\leq j < i$ we have:
  $(\pi\cdot \rho_{|j-1}, \rho(j)) \sat \vfi$. Consider a valuation
  $v$ and an environment $\varepsilon$ such that $(v,\varepsilon)$ is
  compatible with $(\SFnb{\Phi},\pi)$.  Let $v_k$ be the valuation $(v
  +_\varepsilon \rho_{|k-1})$ for $k\in\{0,1,\ldots, i\}$ (where $v_0
  = v$). Clearly $(v_k,\varepsilon) \unrhd ( \SFnb{\Phi},\pi\cdot
  \rho_{|k-1})$, and since $\SFnb{\psi} \subseteq \SFnb{\Phi}$ and
  $\SFnb{\vfi} \subseteq \SFnb{\Phi}$, $(v_k,\varepsilon)$ is
  compatible with $(\SFnb{\psi},\pi\cdot \rho_{|k-1})$ and
  $(\SFnb{\vfi},\pi\cdot \rho_{|k-1})$.  By induction hypothesis, we
  can deduce from (a) and (c) that (a') $(\rho(i),v_i,\varepsilon)\sat
  \overline{\psi}$, and (c') $(\rho(j),v_j,\varepsilon)\sat
  \overline{\vfi}$ for any $j=0,\ldots,i-1$.  Moreover from (b) we can
  deduce: (b') $v_i \sat \overline{C}$.  Thus $(q,v,\varepsilon) \sat
  \Ex \overline {\vfi} \U (\overline{C} \et \overline{\psi})$.

  Conversely, assume $(q,v,\varepsilon) \sat \Ex \overline {\vfi} \U
  (\overline{C} \et \overline{\psi})$. Then there exists
  $\rho\in\Exec(q)$ and $i\geq 0$ such that (a) $(q,v+_\varepsilon
  \rho_{|i-1},\varepsilon) \sat \overline{\psi}$, (b) $v+_\varepsilon
  \rho_{|i-1} \sat \overline{C}$ and (c) for all $0\leq j< i$ we have:
  $(q,v+_\varepsilon \rho_{|j-1},\varepsilon) \sat \overline{\vfi}$.
  Now consider a prefix $\pi$ such that $\pi\cdot q\in\Execf(\calS)$
  and $(v,\varepsilon) \unrhd (\SFnb{\Phi},\pi)$. By induction
  hypothesis, we have: $(\pi\cdot \rho_{|i-1},\rho(i)) \sat \psi$ and
  $(\pi\cdot \rho_{|j-1},\rho(j)) \sat \vfi$ for any
  $j=0,\ldots,i-1$. Hence $(\pi,q) \sat \Ex \vfi \Usc{C} \psi$.

  The case $\Phi \egdef \All \vfi \Usc{C} \psi$ is treated similarly.

  Let now $\Phi \egdef \Now \vfi$, and assume $\SFnb{\vfi} =
  \{\psi_1,\ldots,\psi_k\}$. Let $\varepsilon$ be an environment such
  that $\SFnb{\Phi} \subseteq \dom(\varepsilon)$.  Then for any
  valuation $v_0$ that assigns $0$ to every $\psi_i$,
  \[
  (\pi,q) \sat \Now \vfi \iff (\epsilon,q) \sat \vfi \iff
  (q,v_0,\varepsilon) \sat \overline{\vfi}.
  \]
  This is equivalent to $(q,v,\varepsilon) \sat
  z_{\psi_1}[\overline{\psi_1}] . \ldots z_{\psi_k}[\overline{\psi_k}]
  . \overline{\vfi}$ for any valuation $v$ such that $(v,\varepsilon)
  \unrhd (\SFnb{\Phi},\pi)$. \qedhere
\end{proof}

In fact, $\CCTLc$ can be seen as a variant of $\CCTLv$ where only a
\emph{global} reset operator is available, whose effect corresponds to
the $\Now$ modality in $\CCTLc$. A direct consequence is:

\begin{prop}
  The model checking problem for $\CCTLcb$ is in \PSPACE.
\end{prop}

\noindent
This implies the following corollary:

\begin{cor}
  The model checking problem for all $\CCTLc$ variants up to $\CCTLcb$
  is \PSPACE-complete.
\end{cor}

\noindent
Again, as soon as diagonal constraints are allowed model checking
becomes undecidable:

\begin{thm}
  Model checking $\CCTLcpma$ is undecidable.
\end{thm}

\begin{proof}
  The proof is based on the same technique as that of
  Theorem~\ref{thm:mc-BCS2}.  Consider a two-counter machine $\calM$
  with counters $C$ and $D$ and $n$ instructions.  We define a Kripke
  structure $\calS_\calM = \tuple{Q,R,\ell}$ where $Q \egdef \{q_1,
  \ldots, q_n\} \cup \{r_i, s_i \mid \instr_i \egdef \tuple{\texttt{if
      \ldots}}\}$. The transition relation is defined as follows:
  \begin{iteMize}{$\bullet$}
  \item if $\instr_i \egdef \tuple{X \texttt{++}, j}$, then
    $(q_i,q_j) \in R$ ; and
  \item if $\instr_i \egdef \tuple{\texttt{if} X\texttt{=0 then}\
      j \ \texttt{else } X \texttt{--},\ k}$, then $(q_i,r_i)$,
    $(r_i,q_k)$, $(q_i,s_i)$, $(s_i,q_j)$ in $R$.
  \end{iteMize}
  The labeling $\ell$ is defined over the set $\{\halt\} \,\cup\,
  \bigcup_{X \in \{\texttt{C},\texttt{D}\}} \{X^+, X^-, X^0\}$ as
  $\ell(q_i) \egdef \{X^+\}$ if $\instr_i$ is an increment of $X$,
  $\ell(r_i) \egdef \{X^-\}$ and $\ell(s_i) \egdef \{X^0\}$ if
  $\instr_i$ is a conditional decrement of $X$ and $\ell(q_i) \egdef
  \{\halt\}$ if $\instr_i$ is the halting instruction. One can show
  that there exists a divergent run iff $q_1$ satisfies the formula
  $\Phi_{\calM}$ defined as follows:
  \[
  \EG \Big[ \non \halt \et \ET_{X \in \{\texttt{C},\texttt{D}\}} \Big(
  \big( X^0 \:\impl \: (\nb X^+ = \nb X^-) \big) \et \big( X^- \:\impl
  \: (\nb X^+ > \nb X^-) \big) \Big) \Big]
  \]
  Note that we do not use $\Now$ to prove undecidability.
\end{proof}

Using the same techniques as previously, we obtain the following
results for satisfiability:

\begin{thm}
  The satisfiability problem for all variants of $\CCTLc$ from
  $\CCTLca$ up to $\CCTLcb$ is 2EXPTIME-complete, and becomes
  undecidable for $\CCTLcpma$.
\end{thm}

\section{Conclusion}

In several cases (particularly $\CCTLv$ and thus also $\CCTLb$ and
$\CCTLcb$), the logics we introduce are not more expressive than \CTL
but can concisely express properties which would be difficult to write
in that logic. In particular, even the fragment $\CCTLa$, as well as
$\CCTLb$ with unary-encoded coefficients, can yield exponentially more
succinct formulas than \CTL.

In terms of algorithmic complexity, even though $\CCTLpma$ is strictly
more expressive than \CTL, its model-checking remains polynomial. The
introduction of either coefficients or Boolean combinations increases
the complexity to $\DDP$, while the interplay between Boolean
connectives and possibly negative coefficients yields
undecidability. Similarly, satisfiability is 2-\EXPTIME-complete for
all classes without negative coefficients (when it is simply
\EXPTIME-complete for \CTL \cite{eme85}), and undecidable for all
above classes.
All complexity results are summarized in Figure \ref{global}. 

Further work on \CCTL will include completing the study of
succinctness of its fragments with respect to each other and to other
logics, looking for an upper complexity bound for the model-checking
of $\CCTLpm$, as well as investigating new kinds of constraints. We
also wish to pursue the work described in this article and in
\cite{lmp10ltl} by investigating counting extensions of other temporal
logics (for instance with past operators) as well as $\mu$-calculus.

\subsection*{Acknowledgements} The authors would like to thank the
anonymous referees for their very accurate and helpful comments.

\usetikzlibrary{calc, arrows}
\usetikzlibrary{backgrounds}

\begin{figure}[h]
  \scriptsize
  \subfloat[\CCTL model checking]{
    \begin{tikzpicture}[baseline=(current bounding box.west), scale=.8, node distance=.8cm and 1.6cm, on
      grid]
     \node (c1)     {$\CCTLa$};
      \node (c3)  [below right=of c1]    {$\CCTLpma$};
      \node (ac1) [below left=of c3]     {$\CCTL$};
      \node (ac3) [below right=of ac1]     {$\CCTLpm$};
     \node (bc1)  [above right=of c3]   {$\CCTLba$};
      \node (bc3)  [below right=of bc1]   {$\CCTLbpma$};
      \node (bac1) [below left=of bc3]   {$\CCTLb$};
      \node (bac3)  [below right=of bac1]   {$\CCTLbpm$};

      \path[->] (c1) edge (ac1) edge (bc1) edge (c3)
            (c3) edge (ac3) edge (bc3)
            (ac1) edge (ac3) edge (bac1)
            (ac3) edge (bac3)
            (bc1) edge (bc3) edge (bac1)
            (bc3) edge (bac3)
            (bac1) edge (bac3);

      \begin{pgfonlayer}{background}
        \begin{scope}[fill=black!5,rounded corners=3mm,draw=black!50]
          \filldraw ($(c1)+(-.8,.4)$) -| node[near start,above]
          {\P-complete} ($(c3)+(.8,-.4)$) -| ($(c1)+(-.8,-.5)$) --
          cycle;

          \filldraw ($(bc1)+(-1,.4)$) -- node[above] {\DDP-complete}
          ($(bc1)+(.9,.4)$) -- ($(bac1)+(.9,-.4)$) --
          ($(ac1)+(-.8,-.4)$) -- ($(ac1)+(-.8,.4)$) --
          ($(bac1)+(-1,.4)$) -- cycle;

          \filldraw ($(bc3)+(.9,.4)$) -- node[near start,above]
          {undec.} ($(bc3)+(-.9,.4)$) -- ($(bac3)+(-.9,-.4)$)
          -- ($(bac3)+(.9,-.4)$);

          \filldraw[dashed] ($(ac3)+(-.8,.4)$) rectangle
          ($(ac3)+(.8,-.4)$);
          \node[below=1mm of ac3.south] {{\sf EXPTIME}, \DDP-hard}  ;
        \end{scope}
      \end{pgfonlayer}
    \end{tikzpicture}}
    \hfil
    \subfloat[\CCTLc / \CCTLv model checking]{
    \begin{tikzpicture}[baseline=(current bounding box.west), scale=.8, node distance=.8cm and 1.6cm, on grid]
     \node (bc1)   {$\CCTLcba$};
      \node (c1)  [ left=of bc1]     {$\CCTLca$};
      \node (a2)  [below=of bc1]   {~};
      \node (c3)  [below right=of bc1]    {$\CCTLcpma$};
      \node (bc3)  [right=of c3]   {$\CCTLcbpma$};
      \node (bac1) [below =of a2]   {$\CCTLcb$};
      \node (ac1) [ left=of bac1]     {$\CCTLc$};
      \node (ac3) [below right=of bac1]     {$\CCTLcpm$};
      \node (bac3)  [right=of ac3]   {$\CCTLcbpm$};
      \node (a3)  [below=of bac1]   {~};
      \node (cctlv)  [below=of a3]   {$\CCTLv$};

      \path[->] (c1) edge (ac1) edge (bc1) edge (c3)
            (c3) edge (ac3) edge (bc3)
            (ac1) edge (ac3) edge (bac1)
            (ac3) edge (bac3)
            (bc1) edge (bc3) edge (bac1)
            (bc3) edge (bac3)
            (bac1) edge (bac3)
            (bac1) edge (cctlv);

      \begin{pgfonlayer}{background}
        \begin{scope}[fill=black!5,rounded corners=3mm,draw=black!50]
          \filldraw ($(c1)+(-.8,.4)$) -| node[near start,above]
          {\PSPACE-complete} ($(cctlv)+(.9,-.4)$) -| ($(c1)+(-.8,-1)$) --
          cycle;

          \filldraw ($(bc3)+(.9,.4)$) -- node[near start,above]
          {undecidable} ($(c3)+(-.9,.4)$) -- ($(ac3)+(-.9,-.4)$)
          -- ($(bac3)+(.9,-.4)$);
        \end{scope}
      \end{pgfonlayer}
    \end{tikzpicture}}

    \bigskip

    \subfloat[\CCTL satisfiability]{
    \begin{tikzpicture}[baseline=(current bounding box.west), scale=.8, node distance=.8cm and 1.6cm, on grid]
     \node (bc1)   {$\CCTLba$};
      \node (c1)  [ left=of bc1]     {$\CCTLa$};
      \node (a2)  [below=of bc1]   {~};
      \node (c3)  [below right=of bc1]    {$\CCTLpma$};
      \node (bc3)  [right=of c3]   {$\CCTLbpma$};
      \node (bac1) [below =of a2]   {$\CCTLb$};
      \node (ac1) [ left=of bac1]     {$\CCTL$};
      \node (ac3) [below right=of bac1]     {$\CCTLpm$};
      \node (bac3)  [right=of ac3]   {$\CCTLbpm$};

      \path[->] (c1) edge (ac1) edge (bc1) edge (c3)
            (c3) edge (ac3) edge (bc3)
            (ac1) edge (ac3) edge (bac1)
            (ac3) edge (bac3)
            (bc1) edge (bc3) edge (bac1)
            (bc3) edge (bac3)
            (bac1) edge (bac3);

      \begin{pgfonlayer}{background}
        \begin{scope}[fill=black!5,rounded corners=3mm,draw=black!50]
          \filldraw ($(c1)+(-.8,.4)$) -| node[near start,above]
          {2-\EXPTIME-complete} ($(bac1)+(.9,-.4)$) -| ($(c1)+(-.8,-1)$) --
          cycle;

          \filldraw ($(bc3)+(.9,.4)$) -- node[near start,above]
          {undecidable} ($(c3)+(-.9,.4)$) -- ($(ac3)+(-.9,-.4)$)
          -- ($(bac3)+(.9,-.4)$);
        \end{scope}
      \end{pgfonlayer}
    \end{tikzpicture}}
    \hfil
    \subfloat[\CCTLc / \CCTLv satisfiability]{
    \begin{tikzpicture}[baseline=(current bounding box.west), scale=.8, node distance=.8cm and 1.6cm, on grid]
      \node (bc1)   {$\CCTLcba$};
      \node (c1)  [ left=of bc1]     {$\CCTLca$};
      \node (a2)  [below=of bc1]   {~};
      \node (c3)  [below right=of bc1]    {$\CCTLcpma$};
      \node (bc3)  [right=of c3]   {$\CCTLcbpma$};
      \node (bac1) [below =of a2]   {$\CCTLcb$};
      \node (ac1) [ left=of bac1]     {$\CCTLc$};
      \node (ac3) [below right=of bac1]     {$\CCTLcpm$};
      \node (bac3)  [right=of ac3]   {$\CCTLcbpm$};
      \node (a3)  [below=of bac1]   {~};
      \node (cctlv)  [below=of a3]   {$\CCTLv$};

      \path[->] (c1) edge (ac1) edge (bc1) edge (c3)
            (c3) edge (ac3) edge (bc3)
            (ac1) edge (ac3) edge (bac1)
            (ac3) edge (bac3)
            (bc1) edge (bc3) edge (bac1)
            (bc3) edge (bac3)
            (bac1) edge (bac3)
            (bac1) edge (cctlv);

      \begin{pgfonlayer}{background}
        \begin{scope}[fill=black!5,rounded corners=3mm,draw=black!50]
          \filldraw ($(c1)+(-.8,.4)$) -| node[near start,above]
          {2-\EXPTIME-complete} ($(cctlv)+(.9,-.4)$) -| ($(c1)+(-.8,-1)$) --
          cycle;

          \filldraw ($(bc3)+(.9,.4)$) -- node[near start,above]
          {undecidable} ($(c3)+(-.9,.4)$) -- ($(ac3)+(-.9,-.4)$)
          -- ($(bac3)+(.9,-.4)$);
        \end{scope}
      \end{pgfonlayer}
    \end{tikzpicture}}
  \caption{Summary of model-checking and satisfiability complexity
    results. Arrows indicate syntactic inclusion or straight-forward
    linear translations (case of $\CCTLv$).}
   \label{global}
 \end{figure}
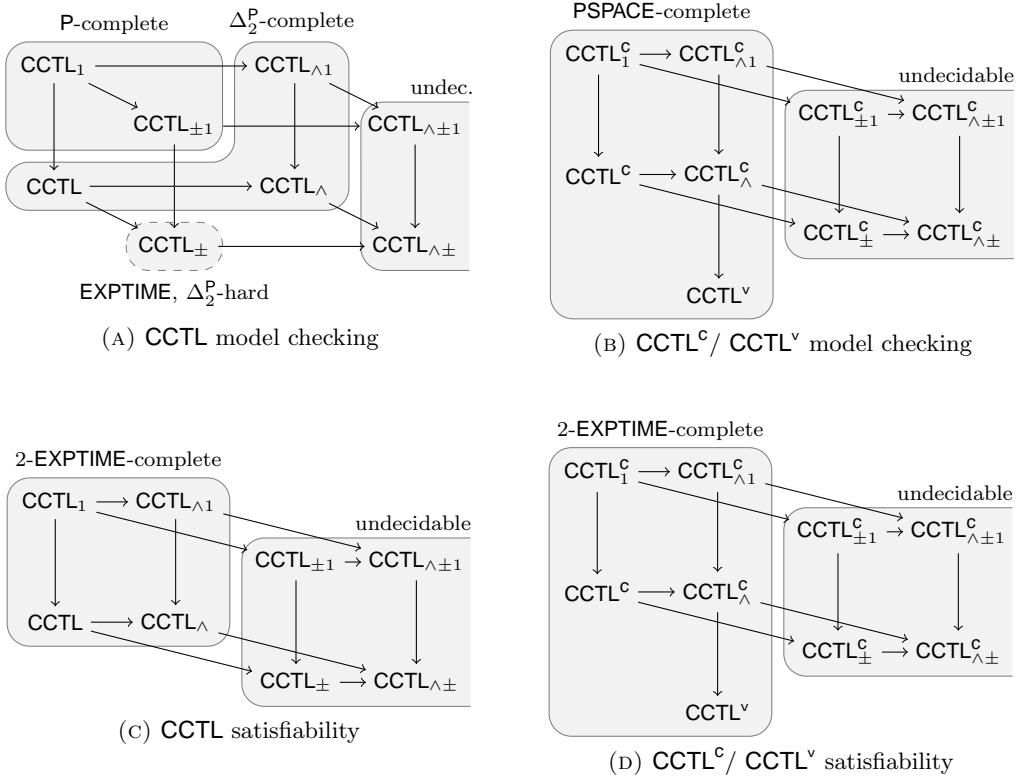

\bibliographystyle{alpha}
\bibliography{local}

\end{document}